\documentclass[12pt,a4paper]{article}
\usepackage{amssymb}

\usepackage{amsmath}
\usepackage{amsthm}
\usepackage{leftidx}
\usepackage{accents}

\usepackage{graphicx}

\newcommand{\re}{\mathrm{Re}}
\newcommand{\im}{\mathrm{Im}}

\newtheorem{lemma}{Lemma}
\newtheorem{theorem}{Theorem}
\newtheorem{definition}{Definition}

\newcommand{\ox}{\otimes}

\newcommand{\cA}{\mathcal A}

\newcommand{\cF}{\mathcal F}

\newcommand{\cH}{\mathcal H}
\newcommand{\cK}{\mathcal K}
\newcommand{\cL}{\mathcal L}
\newcommand{\cP}{\mathcal P}
\newcommand{\cQ}{\mathcal Q}
\newcommand{\cR}{\mathcal R}
\newcommand{\cS}{\mathcal S}
\newcommand{\cT}{\mathcal T}
\newcommand{\cU}{\mathcal U}

\newcommand{\inv}{\mathrm{SI}}
\newcommand{\HH}{\mathrm{HH}}
\newcommand{\trivial}{\mathrm{triv}}

\newcommand{\NB}{\mathit{NB}}

\newcommand{\Poi}{\mathit{Poi}}
\newcommand{\Nor}{\mathit{N}}

\newcommand{\bT}{\mathbb T}
\newcommand{\bN}{\mathbb N}
\newcommand{\bR}{\mathbb R}
\newcommand{\bC}{\mathbb C}
\newcommand{\bK}{\mathbb K}

\newcommand{\Mat}{\mathrm{Mat}}
\newcommand{\Ant}{\mathrm{Ant}}
\newcommand{\Sym}{\mathrm{Sym}}
\newcommand{\Sqz}{\mathrm{Sqz}}
\newcommand{\SO}{\mathrm{SO}}
\newcommand{\GL}{\mathrm{GL}}

\newcommand{\vol}{\mathrm{vol}}

\newcommand{\Null}{\mathrm{Null}}
\newcommand{\Tr}{\mathrm{Tr}}
\newcommand{\tr}{\mathrm{tr}}

\newcommand*{\KET}[1]{\left|#1\right\rangle}
\newcommand*{\BRA}[1]{\left\langle#1\right|}

\newcommand*{\Ket}[1]{\big|#1\big\rangle}
\newcommand*{\Bra}[1]{\big\langle#1\big|}
\newcommand*{\ket}[1]{|#1\rangle}
\newcommand*{\bra}[1]{\langle#1|}

\newcommand*{\dhat}[1]{\underaccent{\tilde}{\hat#1}}
\newcommand*{\dvec}[1]{\underaccent{\tilde}{\vec#1}}
\newcommand*{\dtilde}[1]{\underaccent{\tilde}{#1}}
\newcommand*{\dbar}[1]{\underaccent{\bar}{#1}}

\newcommand*{\donotdisplay}[1]{}

\title
{
	A squeezing invariant measurement 
	to test displacement 
	of 
	quantum Gaussian states 
}
\date{}
\author{Yoshiyuki Tsuda}

\begin{document}
\maketitle

\begin{abstract}
	We consider 
	a hypothesis testing problem 
	for displacement parameters of 
	$n$ independent copies of 
	an $m$-mode squeezed 
	quantum Gaussian state 
	whose mixture parameter is known. 
	Given $n\ge2$, 
	we construct 
	a quantum measurement 
	as a test 
	using 
	an observable 
	which is invariant by 
	$n$-fold tensor product of 
	any $m$-mode squeezing operator. 
	For a pure state case, 
	we calculate 
	the type II error probability 
	of this test. 
	We compare 
	this test 
	with 
	a Hotelling's $T$-squared test 
	which is based on heterodyne measurements. 
\end{abstract}

\section{Introduction}

In quantum hypothesis testing
\cite{helstrom}, 
it is important to test 
whether 
the displacement parameter 
of 
several independent copies of 
a quantum Gaussian state 
\cite{holevo}
is zero or not. 
Kumagai and Hayashi 
\cite{kumagai} 
have studied this problem 
in a situation that 
the mixture parameter is unknown 
and that 
the state is not squeezed. 
They constructed 
a Positive Operator Valued Measurement 
(POVM) 
as a test 
which is invariant by 
some unitary actions.
Using the invariance, 
they proved that 
their test is optimal 
in a minimax criterion.
On the other hand, 
if the states 
are squeezed 
by an unknown squeezing action, 
then, 
even if 
the mixture parameter 
is known, 
we have not obtained 
an optimal test. 
For example,
we will need 
such a test 
when 
we should check  
a displacement channel 
by setting a squeezed probe 
which is 
generated by 
a squeezing device which is 
not completely controlled. 

One of reasonable measurements 
for this problem 
might be 
the Hotelling's $T$-squared test 
\cite{hotelling}
using data from 
heterodyne measurements 
\cite{yuen}. 
We call 
this test 
the Heterodyne-Hotelling (HH) test. 
Let 
$\hat\rho_{\vec\theta,N}$ 
be an $m$-mode non-squeezed 
quantum Gaussian state 
parameterized by 
a displacement vector 
$\vec\theta\in\bC^m$ 
and by 
a mixture value 
$N\ge0$. 
Let 
$\hat S_\eta$ be 
an $m$-mode squeezing operator 
parameterized by 
a matrix 
$\eta$.
The $m$-mode squeezed 
quantum Gaussian state is 
defined by 
$\hat\rho_{\vec\theta,\eta,N}
=\hat S_\eta\hat\rho_{\vec\theta,N}\hat S_\eta^*$. 
If there are 
$n$ independent copies of 
$\hat\rho_{\vec\theta,\eta,N}$,
and 
if $mn$ independent 
heterodyne measurements 
are applied to 
$\hat\rho_{\vec\theta,\eta,N}^{\ox n}$, 
then 
we obtain 
$n$ independent random vectors 
obeying 
a common 
$2m$-dimensional 
normal distribution. 
If $n\ge2m+1$, 
then 
the sample covariance matrix 
$\bar\Sigma_n$ is invertible 
by probability one, 
and 
the Hotelling's $T$-squared statistic 
$T^2$
can be defined using $\bar\Sigma_n^{-1}$. 
Let 
$\mu=2m$, and 
let 
$\nu=n-\mu$.
If 
$\vec\theta$ 
is 
the zero vector 
$\vec0_m\in\bC^m$, 
then 
$F_\HH
=(n-1)^{-1}(\nu/\mu)T^2$ 
obeys 
$F_{\mu,\nu}$,
the central $F$ distribution 
with 
$\mu$ and $\nu$
degrees of freedom.
(See \cite{anderson}.) 
Choose a constant 
$\alpha$ 
as 
a level (of significance).
(See \cite{lehman_romano}.)
Define a critical point 
$c$ 
as 
a solution to the equation 
$\Pr\{F_\HH>c\mid\vec\theta=\vec0_m\}
=\alpha$. 
The HH test 
$T^\HH_\alpha$
of level $\alpha$
is 
a decision rule by which 
$\vec\theta=\vec0_m$ is accepted 
if $F_\HH\le c$ is observed. 
Even if 
$N$
is unknown, 
$T^\HH_\alpha$ can be defined. 
However, 
if $n\le2m$,
then 
$T^\HH_\alpha$ 
can not be defined 
for the sake of 
non-invertibility of 
$\bar\Sigma_n$. 
Moreover, 
even if $N$ is zero, 
$T^\HH_0$ 
is a trivial test 
in a sense that 
$c$
is infinity.
Furthermore, 
in a minimax criterion, 
$T^\HH_\alpha$ 
is not optimal.

In this paper, 
assuming 
$N$ is known 
and 
$\eta$ is unknown, 
we propose a new test. 
Since the new test is 
invariant by 
the action of 
$\hat S_\eta^{\ox n}$, 
it is said to be 
Squeezing Invariant (SI).
The SI test 
of level $\alpha$
is denoted by 
$T^\inv_\alpha$. 
If 
$N$ is unknown, 
then 
$T^\inv_\alpha$ is not defined. 
However, 
for the following three reasons, 
$T^\inv_\alpha$ 
is 
superior to 
$T^\HH_\alpha$. 
First, 
$n\ge2$ is enough 
to define 
$T^\inv_\alpha$. 
Second, 
if $N=0$, 
then 
$T^\inv_0$ 
is not a trivial test. 
Third, 
if $N=0$,
then, 
in a minimax criterion, 
$T^\inv_\alpha$ 
dominates 
$T^\HH_\alpha$.

In Sec. \ref{sec_setup},
we will setup the problem 
by defining 
words and symbols.
In Sec. \ref{sec_construction},
we will construct 
$T^\HH_\alpha$ and $T^\inv_\alpha$.
In Sec. \ref{sec_results},
we will give six theorems 
and a numerical comparison.
In Sec. \ref{sec_proofs},
we will give proofs of the theorems.

\section{Setups}\label{sec_setup}

We define words  and symbols.

\subsection{What is quantum hypothesis testing?}

Let $\cH$
be a Hilbert space. 
For $f\in\cH$,
let $f^*$
be the dual vector,
and let 
$\|f\|
=\sqrt{f^*f}$
be the norm.
Let 
$\cL(\cH)$
be the set of 
linear operators 
on $\cH$. 
Let 
$\hat I_\cH\in\cL(\cH)$
be the identity.
Let 
$X^*$ be the adjoint of 
$X\in\cL(\cH)$.
Let 
$\cU(\cH)\subset\cL(\cH)$
be the set of unitary operators.
If $X\in\cL(\cH)$ is positive,
then we write $X\ge0$.
Let 
$\Tr[X]$ be the trace of 
$X\in\cL(\cH)$.
The set of density operators is given by 
$\cS(\cH)
=\{\hat\rho\in\cL(\cH)\mid\hat\rho\ge0,\ \Tr[\hat\rho]=1\}$.

Consider 
a quantum system described by 
$\cH$ whose state 
$\hat\rho\in\cS(\cH)$ 
is unknown. 
Let 
$\cS_0$ and $\cS_1$ 
be subsets of 
$\cS(\cH)$,
where 
the intersection 
$\cS_0\cap\cS_1$ 
is empty.
Assume that 
either 
$\hat\rho\in\cS_0$ or 
$\hat\rho\in\cS_1$ is true.
When we need to accept 
the null hypothesis 
$H_0:\hat\rho\in\cS_0$ 
or the alternative hypothesis 
$H_1:\hat\rho\in\cS_1$,
it is said that 
we test 
\begin{align}
	H_0:\hat\rho\in\cS_0
	\mbox{ versus }
	H_1:\hat\rho\in\cS_1
	.
	\label{eq_general_hypo}
\end{align}
Let 
$\hat\Pi_0$ and $\hat\Pi_1$
be positive operators satisfying 
$\hat\Pi_0+\hat\Pi_1=\hat I_\cH$.
A test for 
(\ref{eq_general_hypo})
is 
a two-valued POVM 
$\{\hat\Pi_0,\hat\Pi_1\}$
by which 
$H_k$ is accepted 
if $\hat\Pi_k$ is observed.

Let 
$T$
be a test with POVM 
$\{\hat\Pi_0,\hat\Pi_1\}$.
There are two types of mistakes 
caused by $T$.
The type I error is 
acceptance of $H_1$
while $H_0$ is true.
The type II error is 
acceptance of $H_0$ 
while $H_1$ is true.
Hence, 
the type I error probability is 
defined by 
$\alpha_{\hat\rho}[T]=\Tr[\hat\rho\hat\Pi_1]$
as a function of 
$\hat\rho\in\cS_0$, 
and
the type II error probability is 
defined by 
$\beta_{\hat\rho}[T]=\Tr[\hat\rho\hat\Pi_0]$
as a function of 
$\hat\rho\in\cS_1$.
A level (of significance) is 
an upper bound for $\alpha_{\hat\rho}[T]$. 
If $\sup_{\hat\rho\in\cS_0}\alpha_{\hat\rho}[T]\le\alpha$ 
holds, 
then 
$T$ is called 
a test of level $\alpha$.
If $\alpha$ is small,
and if $H_1$ is accepted,
then 
one may be confident that 
$H_1$ is really true 
because the risk of type I error is negligible.
This method was 
proposed by 
\cite{helstrom}
as a generalization of 
the classical statistical hypothesis testing theory, 
which is described in 
\cite{lehman_romano}.

For any $\alpha$
with $0\le\alpha\le1$,
there exists a test of level $\alpha$. 
Let 
$T^\trivial_\alpha$ 
be a test 
defined by 
$\hat\Pi_0=(1-\alpha)\hat I_\cH$.
For any $\hat\rho\in\cS_0$,
it holds that 
$\alpha_{\hat\rho}[T^\trivial_\alpha]=\alpha$.
Hence, 
$T^\trivial_\alpha$ 
is a test of level $\alpha$,
and is called 
a trivial test 
of level $\alpha$. 

If $T_1$ and $T_2$ are tests 
of a common level $\alpha$, 
then they are compared by 
the type II error probabilities. 
If 
\begin{align}
	\beta_{\hat\rho}[T_1]
	\le
	\beta_{\hat\rho}[T_2]
	\
	({^\forall}\hat\rho\in\cS_1)
	\
	\mbox{ and }
	\
	\beta_{\hat\rho}[T_1]
	<
	\beta_{\hat\rho}[T_2]
	\
	({^\exists}\hat\rho\in\cS_1)
	\label{eq_t1_vs_t2}
\end{align}
hold, 
then we conclude that 
$T_1$ dominates $T_2$. 
If a test $T$ 
dominates 
$T^\trivial_\alpha$,
then there exists 
$\hat\rho\in\cS_1$ such that 
$\beta_{\hat\rho}[T]<1-\alpha$.

In many cases, 
however, 
the condition 
(\ref{eq_t1_vs_t2})
is so strict that 
we can not 
complete 
the comparison. 
Hence, 
we use 
a minimax criterion 
instead of 
(\ref{eq_t1_vs_t2}).
Consider a case where 
$\hat\rho\in\cS_0\cup\cS_1$ 
is parameterized by 
$\theta\in\Theta$ and $\xi\in\Xi$ as 
$\hat\rho_{\theta,\xi}$. 
If there exists 
$\Theta_1\subset\Theta$ such that 
$\cS_1=\{\hat\rho_{\theta,\xi}\mid\theta\in\Theta_1,\ \xi\in\Xi\}$, 
then 
we are not interested in 
the true value of $\xi$. 
In such a case,
$\theta$ is called 
the parameter of interest,
and $\xi$ is called 
the parameter of nuisance.
(In this sense,
the displacement 
is the parameter of interest,
and the squeezing 
is the parameter of nuisance.)
The condition 
(\ref{eq_t1_vs_t2}) 
is modified as 
\begin{align}
	&
	\sup_{\xi\in\Xi}\beta_{\hat\rho_{\theta,\xi}}[T_1]
	\le
	\sup_{\xi\in\Xi}\beta_{\hat\rho_{\theta,\xi}}[T_2]
	\
	({^\forall\theta\in\Theta_1})
	\nonumber
	\\
	\mbox{ and }
	&
	\sup_{\xi\in\Xi}\beta_{\hat\rho_{\theta,\xi}}[T_1]
	<
	\sup_{\xi\in\Xi}\beta_{\hat\rho_{\theta,\xi}}[T_2]
	\
	({^\exists\theta\in\Theta_1})
	.
	\label{eq_minimax}
\end{align}
If $T_1$ and $T_2$ 
of a common level $\alpha$
satisfy 
the condition
(\ref{eq_minimax}), 
then 
we conclude that 
$T_1$ domintates $T_2$ 
in the minimax criterion. 
If a test $T$ of level $\alpha$
satisfies 
$\sup_{\xi\in\Xi}\beta_{\hat\rho_{\theta,\xi}}[T]=1-\alpha$
$({^\forall}\theta\in\Theta_1)$,
then, 
in the minimax criterion, 
$T$ is no better than 
$T^\trivial_\alpha$.

Kumagai and Hayashi 
\cite{kumagai} 
studied 
a theory of 
the minimax criterion 
in quantum hypothesis testing,
and they showed that
an optimality 
in the minimax criterion 
is concerned with 
unitary invariance.
Let $\theta\in\Theta$ 
be the parameter of interest,
and
let $\xi\in\Xi$ 
be the parameter of nuisance.
Let $U:\Xi\to\cU(\cH)$ be a map 
given as 
$\xi\mapsto U_\xi$. 
Assume that 
there exists 
$\xi_0\in\Xi$ 
such that 
$\hat\rho_{\theta,\xi}
=U_\xi\hat\rho_{\theta,\xi_0} U_\xi^*$ 
holds 
for any $\theta\in\Theta$ 
and 
for any $\xi\in\Xi$.
A test $T$ 
whose POVM 
$\{\hat\Pi_0,\hat\Pi_1\}$ 
satisfies 
$U_\xi^*\hat\Pi_0U_\xi=\hat\Pi_0$ 
$({^\forall}\xi\in\Xi)$
is said to be invariant by $U$. 
If $T$ is invariant by $U$,
then it holds that 
\begin{align}
	\beta_{\hat\rho_{\theta,\xi}}[T]
	=
	\Tr[\hat\rho_{\theta,\xi_0}\hat\Pi_0]
	\quad
	({^\forall}\theta\in\Theta_1,\
	{^\forall}\xi\in\Xi)
	\label{eq_necessary_inv}
	.
\end{align}
Kumagai and Hayashi \cite{kumagai}
proved that 
a test is invariant by such $U$
if the test is optimal 
in the minimax criterion.

We do not prove optimality 
of 
$T^\inv_\alpha$. 
However,
we prove that 
$T^\inv_\alpha$
is SI, 
and that, 
if the mixture 
$N$ is zero,
then 
$T^\inv_\alpha$
dominates 
$T^\trivial_\alpha$ 
in the minimax criterion.
Moreover,
we prove that 
$T^\HH_\alpha$
does not satisfy 
(\ref{eq_necessary_inv}), 
which is 
a necessary condition 
for $T$ to be SI.
Furthermore,
we prove that 
$T^\HH_\alpha$
is no better than 
$T^\trivial_\alpha$
in the minimax criterion.
Hence,
$T^\inv_\alpha$ 
dominates 
$T^\HH_\alpha$
in the minimax criterion.

\subsection{Notations of sets of matrices}

To parameterize 
multi-mode squeezing operators, 
we use several sets of matrices.
Let 
$\bN$
be the set of positive integers..
For any $m,n\in\bN$,
let 
$\Mat_\bK^{m,n}$
be the set of 
$m$-by-$n$ matrices 
whose entries belong to 
$\bK$, 
which will be 
$\bC$ 
or 
$\bR$.
For $X\in\Mat_\bC^{m,n}$,
the transpose 
is denoted by 
${^t}X\in\Mat_\bC^{n,m}$,
the entry-wise complex conjugate is 
denoted by 
$\bar X\in\Mat_\bC^{m,n}$,
and 
the adjoint 
${^t}\bar X$
is denoted by 
$X^*
\in\Mat_\bC^{n,m}$.
Let 
$\Mat_\bK^m$ 
be 
$\Mat_\bK^{m,m}$. 
The set of 
anti-hermitian matrices 
is defined by 
$\Ant_\bK^m
=\{A\in\Mat_\bK^m\mid A=-A^*\}$.
The set of symmetric matrices 
is defined by 
$\Sym_\bK^m
=\{S\in\Mat_\bK^m\mid S={^t}S\}$.
Define 
$\Sqz^m
\subset\Mat_\bC^{2m}$
by 
\[
	\Sqz^m
	=
	\left\{
	\begin{pmatrix}
	A&S\\\bar S&\bar A
	\end{pmatrix}
	\ \middle|\
	A\in\Ant_\bC^m
	,\
	S\in\Sym_\bC^m
	\right\}
	.
\]
For 
$\eta\in\Sqz^m$,
the upper-left submatrix $A\in\Ant_\bC^m$ 
is called the anti-hermitian part 
of $\eta$,
and 
the upper-right submatrix 
$S\in\Sym_\bC^m$
is called the symmetric part 
of $\eta$.

\subsection{Multi-mode squeezed quantum Gaussian states}

Let 
$\cH$
be 
$L^2(\bR)$,
the set of 
$\bC$-valued
square-integrable 
functions 
of a real coordinate variable 
$x\in\bR$.
The inner product 
of $f,g\in\cH$ 
is defined by 
$f^*g
=\int_\bR\overline{f(x)}g(x)dx$,
where 
$\bar z$ is the conjugate of
$z\in\bC$.
A single-mode 
electromagnetic field is 
described by 
$\cH$.
(See 
\cite{leonhardt} and \cite{milburn}.)
For $\theta\in\bC$, 
the coherent vector 
$\ket\theta\in\cH$ 
is defined by 
\begin{align*}
	\ket\theta(x)
	=
	\frac
	{e^{-|\theta|^2/2}}
	{\pi^{1/4}}
		e^{
			-x^2/2
			+\sqrt2\theta x
			-\theta^2/2
		}
	.
\end{align*}
Let 
$\bra\theta$ 
be 
$\ket\theta^*$. 
For $\theta\in\bC$,
and for $N\ge0$,
the single-mode non-squeezed 
quantum Gaussian state 
$\hat\rho_{\theta,N}\in\cS(\cH)$
is defined by 
\[
	\hat\rho_{\theta,N}
	=
	\begin{cases}
	\ket\theta\bra\theta & \mbox{if }N=0,\\
	\displaystyle
	\frac{1}{\pi N}
	\iint_{\bR^2}
	e^{-|r+is-\theta|^2/N}
	\ket{r+is}\bra{r+is}
	drds
	&\mbox{if }N>0,
	\end{cases}
\]
where 
$i=\sqrt{-1}$.
(See \cite{holevo}.)
The state is pure 
if $N=0$.

For $m\in\bN$,
an $m$-mode system is described by 
$\cH^{\ox m}$. 
For 
$\vec\theta
={^t}(\theta_1,\theta_2,...,
\allowbreak
\theta_m)
\in\Mat_\bC^{m,1}$, 
and 
for $N\ge0$, 
the $m$-mode non-squeezed 
quantum Gaussian state 
$\hat\rho_{\vec\theta,N}
\in\cS(\cH^{\ox m})$
is defined by 
$\hat\rho_{\vec\theta,N}
=\hat\rho_{\theta_1,N}\ox\hat\rho_{\theta_2,N}\ox\cdots\ox\hat\rho_{\theta_m,N}$.

Define 
$\hat q\in\cL(\cH)$ 
by 
$\hat qf(x)=xf(x)$, 
and 
define 
$\hat p\in\cL(\cH)$ 
by 
$\hat pf(x)=-idf(x)/dx$, 
where 
$i=\sqrt{-1}$. 
They satisfy 
$\hat q=\hat q^*$,
$\hat p=\hat p^*$ 
and 
$[\hat q,\hat p]
=\hat q\hat p-\hat p\hat q
=i\hat I$, 
where 
$\hat I\in\cL(\cH)$ is the identity. 
The annihilation operator 
$\hat a\in\cL(\cH)$
is defined by 
$\hat a=(\hat q+i\hat p)/\sqrt2$. 
It holds that 
$[\hat a,\hat a^*]
=\hat I$, 
and that 
\begin{align}
	\hat a\ket\theta
	=
	\theta\ket\theta
	.
	\label{eq_hat_a_ket_theta}
\end{align}
For 
$i\in\{1,2,...,m\}$,
the $i$-th annihilation operator 
$\hat a_i
\in\cL(\cH^{\ox m})$
is defined by 
$\hat a_i
=\hat I^{\ox(i-1)}\ox\hat a\ox\hat I^{\ox(m-i)}$.
For 
$\eta\in\Sqz^m$,
let 
$A_{i,j}$ 
and 
$S_{i,j}$
be the $(i,j)$-th entries 
of the anti-hermitian part 
and 
of the symmetric part,
respectively,
and let 
\[
	\hat s_\eta
	=
	\sum_{i=1}^m
	\sum_{j=1}^m
	\Big(
	A_{i,j}\hat a_i^*\hat a_j
	+\frac12
	S_{i,j}\hat a_i^*\hat a_j^*
	-\frac12
	\bar S_{i,j}\hat a_i\hat a_j
	\Big)
	.
\]
An $m$-mode squeezing operator 
$\hat S_\eta
\in\cU(\cH^{\ox m})$
is defined by 
$\hat S_\eta
=\exp(\hat s_\eta)$.
The $m$-mode 
squeezed quantum Gaussian state 
$\hat\rho_{\vec\theta,\eta,N}
\in\cS(\cH^{\ox m})$
is defined by 
$\hat\rho_{\vec\theta,\eta,N}
=\hat S_\eta\hat\rho_{\vec\theta,N}\hat S_\eta^*$.

\subsection{Our hypothesis testing problem}

Suppose that 
a quantum state 
of the form 
$\hat\rho_{\vec\theta,\eta,N}^{\ox n}
\in\cS(\cH^{\ox mn})$
is given.
We call 
$m\in\bN$ the mode size,
$n\in\bN$ the sample size,
$\vec\theta\in\Mat_\bC^{m,1}$ the displacement parameter,
$\eta\in\Sqz^m$ the squeezing parameter
and
$N\ge0$ the mixture parameter.
We assume that 
$\vec\theta$
and 
$\eta$ are unknown, 
and that 
$N$ is known.
Our problem is 
to test 
\begin{align}
	H_0:\vec\theta=\vec0_m
	\mbox{ versus }
	H_1:\vec\theta\ne\vec0_m
	,
	\label{eq_hypo}
\end{align}
where 
$\vec0_m\in\Mat_\bC^{m,1}$
is the zero vector.

The squeezing parameter 
$\eta\in\Sqz^m$ 
is a nuisance parameter 
because 
$(\ref{eq_hypo})$ 
does not depend on 
$\eta$. 
Hence, 
the minimax criterion 
(\ref{eq_minimax}) 
is specified by 
$\Theta_1=
\{\vec\theta\in\Mat_\bC^{m,1}\mid\vec\theta\ne\vec0_m\}$ 
and by 
$\Xi=\Sqz^m$. 
For a subspace 
$\cK\subset\cH^{\ox mn}$,
and for 
$\hat L\in\cL(\cH^{\ox mn})$,
let 
$\hat L\cK\subset\cH^{\ox mn}$
be 
$\{\hat Lf\mid f\in\cK\}$.
If 
$\hat S_\eta^{\ox n}\cK=\cK$
holds for any 
$\eta\in\Sqz^m$,
then 
$\cK$ is said to be Squeezing Invariant (SI).
If $\hat L\in\cL(\cH^{\ox mn})$
satisfies 
$\hat L=(\hat S_\eta^{\ox n})^*\hat L\hat S_\eta^{\ox n}$
$({^\forall}\eta\in\Sqz^m)$,
then 
$\hat L$ is said to be 
SI.
A test 
with POVM 
$\{\hat\Pi_0,\hat\Pi_1\}$ 
is said to be SI
if $\hat\Pi_0$ is SI.

\subsection{Definition of $a_{i,j}$}

For 
$i\in\{1,2,...,m\}$
and
for 
$j\in\{1,2,...,n\}$,
the $(i,j)$-th annihilation operator 
$\hat a_{i,j}
\in\cL(\cH^{\ox mn})$
is defined by 
\[
	\hat a_{i,j}
	=
	\hat I^{\ox(jm-m)}\ox
	\hat a_i
	\ox\hat I^{\ox(mn-jm)}
	=
	\hat I^{\ox(i-1+jm-m)}\ox
	\hat a
	\ox\hat I^{\ox(m-i+mn-jm)}
	.
\]

\section{Constructions of $T^\inv_\alpha$ and $T^\HH_\alpha$}\label{sec_construction}

\subsection{Construction of $T^\inv_\alpha$}

Assume that $n\ge2$.
We first 
construct 
an observable 
$\hat T_\inv\in\cL(\cH^{mn})$, 
which is positive and SI.
For 
$j,k\in\{1,2,...,n\}$,
let 
$
	\hat v_{j,k}
	=
	\sum_{i=1}^m
	(\hat a_{i,k}^*\hat a_{i,j}
	-\hat a_{i,j}^*\hat a_{i,k})
$. 
We will show,
in Theorem \ref{th_inv}, 
that 
$\hat v_{j,k}$ 
is SI.
Moreover,
by Lemma \ref{lem:2_realization},
$\hat v_{j,k}$
is unitarily equivalent to 
\begin{align}
	\hat d_{j,k}
	=
	\sqrt{-1}
	\sum_{i=1}^m
	(\hat a_{i,j}^*\hat a_{i,j}-\hat a_{i,k}^*\hat a_{i,k})
	.
	\label{eq_hat_d_j_k}
\end{align}
For $k\in\{1,2,...,n-1\}$,
let 
$
	\hat r_k
	=
	\arctan(\sqrt k)
	\hat v_{k,k+1}
$,
and let 
$\hat R_k
=\exp(\hat r_k)
\in\cU(\cH^{\ox mn})$.
Let 
$\hat R
=\hat R_{n-1}\hat R_{n-2}\cdots\hat R_1
\in\cU(\cH^{\ox mn})$.
Let 
$
	\hat T_\inv
	=
	\sum_{k=1}^{n-1}
	\hat R^*
	\hat v_{k,n}\hat v_{k,n}^*
	\hat R
$.
Because of 
$\hat v_{k,n}\hat v_{k,n}^*\ge0$,
we have 
$\hat T_\inv\ge0$.
Moreover,
since 
$\hat v_{j,k}$ 
is SI,
$\hat T_\inv$ 
is SI. 

Next, 
we construct the SI test 
of level $\alpha\in[0,1]$.
For $t\in\bR$,
define 
a Hilbert subspace 
$\mathcal K_t
\subset\cH^{\ox mn}$
by 
\[
	\mathcal K_t
	=
	\{
	f\in\cH^{\ox mn}
	\mid
	f^*\hat T_\inv f
	\le
	t
	\|f\|^2
	\}
	.
\]
Since 
$\hat T_\inv$ is SI,
$\cK_t$ is SI.
Since 
$\hat T_\inv\ge0$
holds, 
$s<0$ implies 
$\mathcal K_s=\{0\}$.
Let 
$\hat K_t
\in\cL(\cH^{\ox mn})$
be the projection on 
$\mathcal K_t$.
Since 
$\cK_t$
is SI,
$\hat K_t$ is SI.
Hence,
for any 
$\eta\in\Sqz^m$,
it holds that 
$
	\Tr[\hat\rho_{\vec\theta,\eta,N}^{\ox n}\hat K_t]
	=\Tr[\hat\rho_{\vec\theta,N}^{\ox n}\hat K_t]
$.

Assume that $N$
is known.
For 
$\alpha\in[0,1]$, 
let 
$s,t,w\in\bR$
be solutions to 
\[
	\begin{cases}
	1-\alpha
	=
	(1-w)\Tr[\hat\rho_{\vec0_m,N}^{\ox n}
	\hat K_s]
	+
	w\Tr[\hat\rho_{\vec0_m,N}^{\ox n}
	\hat K_t]
	,
	\\
	s<t\mbox{ and }0<w\le1.
	\end{cases}
\]
Let 
$
	\hat\Pi_0
	=
	(1-w)\hat K_s
	+w
	\hat K_t
$,
and let 
$\hat\Pi_1=\hat I^{\ox mn}-\hat\Pi_0$.
The SI test 
$T^\inv_\alpha$ 
is defined by 
the POVM 
$\{\hat\Pi_0,\hat\Pi_1\}$.

\donotdisplay
{
There are three reasons 
why $T_\inv$ is better than the HH test. 
First, 
the test $T_\inv$ 
can be constructed 
if $n\ge2$. 
Second, 
if $N=0$, 
then 
a test of level zero 
can be constructed by 
$\hat\Pi_0=\hat K_0$. 
Third, 
$\sup_{\eta\in\Sqz^m}\beta_{T_\inv}<1$ 
holds for any $\vec\theta\ne\vec0_m$. 
}

\subsection{Construction of $T^\HH_\alpha$}

Let 
$\cF$ be the set 
of Borel subsets of 
$\bR^2$. 
A single-mode heterodyne measurement 
is a POVM 
defined by 
\[
	\cF\ni M
	\mapsto 
	\frac{1}{\pi}
	\iint_M
	\ket{x+iy}\bra{x+iy}
	dxdy
	,
\]
where 
$i=\sqrt{-1}$. 
For $\eta\in\Sqz^m$,
let 
$A\in\Ant_\bC^m$ 
be the anti-hermitian part,
and let 
$S\in\Sym_\bC^m$
be the symmetric part.
Define 
$\vec\mu_{\vec\theta}\in\Mat_\bR^{2m,1}$
and
$G_\eta\in\Mat_\bR^{2m}$
by 
\begin{align}
	\vec\mu_{\vec\theta}
	=
	\begin{pmatrix}
	\re(\vec\theta)\\\im(\vec\theta)
	\end{pmatrix}
	\mbox{ and }
	G_\eta
	=
	\exp
	\begin{pmatrix}
	\re(A)+\re(S)&-\im(A)+\im(S)
	\\
	\im(A)+\im(S)&\re(A)-\re(S)
	\end{pmatrix}
	,
	\label{eq_g_eta}
\end{align}
respectively.
Define 
$\vec\mu_{\vec\theta,\eta}\in\Mat_\bR^{2m,1}$
and 
$\Sigma_{\eta,N}\in\Mat_\bR^{2m}$
by 
\begin{align}
	\vec\mu_{\vec\theta,\eta}=G_\eta\vec\mu_{\vec\theta}
	\mbox{ and }
	\Sigma_{\eta,N}
	=
	\frac{2N+1}{4}
	G_\eta({^t}G_\eta)
	+
	\frac14I_{2m}
	,
	\label{eq_vec_mu_vec_theta_eta_sigma_eta_n}
\end{align}
respectively.
By 
Lemma \ref{lem_multi_normal},
applying $mn$ independent 
single-mode 
heterodyne measurements 
to $\hat\rho_{\vec\theta,\eta,N}^{\ox n}$, 
we obtain 
$n$ independent 
$2m$-dimensional random vectors 
$\vec X_1,\vec X_2,...,\vec X_n$ 
according to 
a common $2m$-dimensional 
normal distribution 
whose mean vector is 
$\vec\mu_{\vec\theta,\eta}$
and 
whose covariance matrix is 
$\Sigma_{\eta,N}$;
say 
$\Nor_{2m}(\vec\mu_{\vec\theta,\eta},\Sigma_{\eta,N})$.
Let 
$\bar X_n$
be the sample mean vector 
$n^{-1}\sum_{j=1}^n\vec X_j$, 
and let 
$\bar\Sigma_n$
be the sample covariance matrix 
$(n-1)^{-1}\sum_{j=1}^n(\vec X_j-\bar X_n)({^t}\vec X_j-{^t}\bar X_n)$. 

Assume that $n\ge2m+1$. 
Then, 
$\bar\Sigma_n$ has the inverse 
$\bar\Sigma_n^{-1}$
by probability one. 
The Hotelling's $T$-squared statistic 
is defined by 
$T^2=n({^t}\bar X_n)\bar\Sigma_n^{-1}\bar X_n$. 
Let 
$\mu=2m$,
$\nu=n-2m$,
$F_\HH
=(n-1)^{-1}(\nu/\mu)T^2$
and 
$
	\lambda
	=
	n({^t}\vec\mu_{\vec\theta,\eta})\Sigma_{\eta,N}^{-1}\vec\mu_{\vec\theta,\eta}
$. 
Then,
$F_\HH$
obeys 
$F_{\mu,\nu;\lambda}$,
the non-central $F$ distribution 
with $\mu$ and $\nu$ degrees of freedom 
and with non-centrality 
$\lambda$. 
The probability density function of 
$F_{\mu,\nu;\lambda}$ is
\begin{align}
	p_\lambda(f)
	=
	\sum_{k=0}^\infty
	\frac
	{e^{-\lambda/2}
	(\lambda/2)^k
	/k!}
	{B\left(k+\mu/2,\nu/2\right)}
	\left(\frac{\mu f}{\mu f+\nu}\right)^{k+\mu/2}
	\left(\frac{\nu}{\mu f+\nu}\right)^{\nu/2}
	\frac1f
	,
	\label{eq_pdf_non_central_f}
\end{align}
where 
$B(x,y)$
is the beta function.
(See \cite{anderson}.)
If $\vec\theta=\vec0_m$, 
then $\lambda=0$ 
and thus 
$F_\HH$ obeys 
$F_{\mu,\nu}=F_{\mu,\nu;0}$
the central $F$ distribution. 
Define 
the critical point $c$ 
as a solution to the equation 
$\int_c^\infty p_0(f)df=\alpha$.
\donotdisplay
{
\[
	\Pr\{F_\HH>c\mid\lambda=0\}
	=
	\frac{1}{B(\mu/2,\nu/2)}
	\int_c^\infty
	\left(\frac{\mu x}{\nu+\mu x}\right)^{\mu/2}
	\left(\frac{\nu}{\nu+\mu x}\right)^{\nu/2}
	\frac1x
	dx
	=\alpha
	.
\]
}
The HH test $T^\HH_\alpha$ 
is a test 
by which 
$H_0$ is accepted 
if $F_\HH\le c$ is observed. 

For 
$\vec z
={^t}(z_1,z_2,...,z_m)
\in\Mat_\bC^{m,1}$,
define 
$\Ket{\vec z}\in\cH^{\ox m}$ 
by 
$\Ket{\vec z}
=\ket{z_1}\ox\ket{z_2}\ox\cdots\ox\ket{z_m}$.
For 
$Z
=(\vec z_1,\vec z_2,...,\vec z_n)
\in\Mat_\bC^{m,n}$,
define 
$\ket Z\in\cH^{\ox mn}$
by 
$\ket Z
=\ket{\vec z_1}\ox\Ket{\vec z_2}\ox\cdots\ox\Ket{\vec z_n}$.
Let $\bra Z$ be $\ket Z^*$.
The POVM 
$\{\hat\Pi_0,\hat\Pi_1\}$ 
for $T^\HH_\alpha$ 
is given by 
\[
	\hat\Pi_0
	=
	\frac{1}{\pi^{mn}}
	\left.
	\overbrace{\int\cdots\int}^{2mn}
	\right._{F_\HH\le c}
	\ket Z\bra Z
	\prod_{i=1}^m\prod_{j=1}^n
	dx_{i,j}dy_{i,j}
	,
\]
where 
$x_{i,j}$ and $y_{i,j}$
are $(i,j)$-th entries of 
$\re(Z)$ and $\im(Z)$,
respectively.

\donotdisplay
{
If $\vec\theta\ne\vec0_m$, 
then 
the probability distribution of 
$f$ 
depends on 
$\eta\in\Sqz^m$. 
Hence, 
$T_\HH$ is 
not SI. 
Even if $N$ is unknown, 
$T_\HH$ can be defined. 
However, 
if $n\le2m$, 
then 
$\bar\Sigma_n^{-1}$ does not exist 
and thus 
$T_\HH$ can not be defined. 
Moreover, 
if $\alpha=0$, 
then $t_0=\infty$ 
and thus 
$T_\HH$ is trivial, 
that is, 
$\hat\Pi_0=\hat I^{\ox mn}$. 
Furthermore, 
if $\vec\theta\ne\vec0_m$, 
then, 
for any $\epsilon>0$, 
there exists 
$\eta\in\Sqz^m$ 
such that 
$\Pr\{T^2>t_0\mid\vec\theta,\eta\}<\alpha+\epsilon$. 
That is to say, 
$\sup_{\eta\in\Sqz^m}\beta_{T_\HH}=1-\alpha$ 
for any $\vec\theta$ 
with $\vec\theta\ne\vec0_m$, 
and thus 
$T_\HH$ 
is one of the worst tests 
in the minimax criterion. 
}

\section{Theorems and a numerical comparison}\label{sec_results}

Let 
$m$
be the mode size,
$n$ be the sample size,
$\vec\theta$
be the displacement parameter,
$\eta$ be the squeezing parameter,
$N$ be the mixture parameter,
and 
$\alpha$ be the level.
The first theorem implies that 
$\hat v_{j,k}\in\cL(\cH^{\ox mn})$,
$\hat R=\hat R_{n-1}\hat R_{n-2}\cdots\hat R_1\in\cU(\cH^{\ox mn})$,
$\hat T_\inv\in\cL(\cH^{\ox mn})$,
$\cK_t\subset\cH^{\ox mn}$,
$\hat K_t\in\cL(\cH^{\ox mn})$
and 
$T^\inv_\alpha$
are all SI.

\begin{theorem}\label{th_inv}
For any $\eta\in\Sqz^m$, 
and for any $j,k\in\{1,2,...,n\}$,
it holds that 
$(\hat S_\eta^{\ox n})^*\hat v_{j,k}\hat S_\eta^{\ox n}=\hat v_{j,k}$.
\end{theorem}
\begin{proof}
See 
Sec. \ref{sec_proof_inv}.
\end{proof}

A vector $f\in\cH^{\ox mn}$
is regarded as a function 
$f:\Mat_\bR^{m,n}\to\bC$,
where 
the $(i,j)$-th entry 
$x_{i,j}$ 
of the matrix variable 
$X\in\Mat_\bR^{m,n}$
is specified by 
$\hat a_{i,j}f
=(x_{i,j}f+\partial f/\partial x_{i,j})/\sqrt2$.
If 
$f(X)=f(Xe^A)$ holds 
for any $A\in\Ant_\bR^n$,
then 
$f(X)$ is said to be 
Mode-wisely Rotationally Invariant
(MwRI).
If $m=1$ and if $n\ge2$,
then 
an MwRI function is a radial function 
of $n$ variables.

\begin{theorem}\label{th_radial}
The Hilbert subspace 
$\cK_0$ 
is the set of 
square-integrable MwRI functions.
\end{theorem}
\begin{proof}
See 
Sec. \ref{sec_proof_radial}.
\end{proof}

Consider the case of $n=2$.
Then,
$\hat T_\inv
=\hat v_{1,2}\hat v_{1,2}^*$
holds 
because 
$\hat R^*\hat v_{1,2}\hat R=\hat v_{1,2}$
holds.
Let $X$ be a random variable 
given by observing $\hat T_\inv$.
Let
$\hat T
=-i\hat v_{1,2}$.
where $i=\sqrt{-1}$. 
It holds that 
$\hat T^*=\hat T$
and that 
$\hat T_\inv=\hat T^2$.
Let $Y$ be a random variable 
given by observing $\hat T$.
For any state, 
the probability distribution of 
$X$ is 
equal to 
that of $Y^2$.

For a random variable $Z$,
the probability distribution 
can be identified by 
the characteristic function 
$\varphi_Z(r)=E[e^{i r Z}]$,
where 
$E[W]$ is the expected value of 
a random variable $W$.
Let 
$\NB_m(p)$ 
be the negative binomial distribution 
whose probability function is 
$
	f(x)
	=
	\begin{pmatrix}m+x-1\\x\end{pmatrix}
	(1-p)^mp^x
$.
If $Z$ obeys $\NB_m(p)$,
then it holds that 
$\varphi_Z(r)=(1-p)^m(1-pe^{ir})^{-m}$.
Let 
$\Poi(\lambda)$ 
be the Poisson distribution 
whose probability function is 
$
	f(x)
	=
	e^{-\lambda}
	\lambda^x/(x!)
$.
If $W$ obeys $\Poi(\lambda)$,
then it holds that 
$\varphi_W(r)=e^{\lambda(e^{ir}-1)}$.
Define 
$\gamma(r)$ 
and 
$\psi_s(r)$
by 
\begin{align}
	\gamma(r)
	=
	\frac
	{1}
	{
	N+1-Ne^{ir}
	}
	\mbox{ and }
	\psi_s(r)
	=
	\exp[\gamma(r)(e^{ir}-1)s^2]
	\label{eq_def_gamma}
\end{align}
\donotdisplay
{
and let 
\begin{align}
	\psi_s(r)
	=
	\begin{cases}
	\exp(
	s^2e^{ir}-s^2
	)
	& \mbox{if }N=0,
	\\
	\exp\left(
	\frac{s^2/N}{N+1-Ne^{ir}}
	-\frac{s^2}{N}
	\right)
	& \mbox{if }N>0.
	\end{cases}
	\label{eq_def_psi}
\end{align}
}

\begin{theorem}\label{th_n2}
(i) 
If $n=2$, 
then it holds that 
\[
	\varphi_Y(r)
	=
	[\gamma(r)\gamma(-r)]^m\psi_{\|\vec\theta\|}(r)\psi_{\|\vec\theta\|}(-r)
	.
\]
\\
(ii) 
Assume that
$F$, $G$, 
$P_k$
and $Q_k$
$(k\in\bN)$
are mutually independent 
random variables, 
where 
$F$ and $G$ obey 
$\NB_m(N/(N+1))$,
and where 
$P_k$ and $Q_k$ obey 
$
	\Poi\left(
	\frac
	{\|\theta\|^2N^{k-1}}
	{(N+1)^{k+1}}
	\right)
$.
If $n=2$,
then 
the probability distribution 
of $Y$ 
is that of 
$
	F-G
	+
	\sum_{k=1}^\infty
	(kP_k-kQ_k)
$.
\end{theorem}
\begin{proof}
See Sec. \ref{sec_n2}.
\end{proof}

The following theorem shows that 
$T^\inv_\alpha$ 
dominates 
$T^\trivial_\alpha$
if $N=0$.

\begin{theorem}\label{th_type2}
If $N=0$, 
then, 
for any $m\ge1$ 
and for any $n\ge2$, 
the type II error probability of 
$T^\inv_\alpha$
is 
\[
	\beta_{\hat\rho_{\vec\theta,\eta,0}^{\ox n}}[T^\inv_\alpha]
	=
	(1-\alpha)
	\frac{e^{-n\|\vec\theta\|^2}}{B\left(\frac{n-1}{2},\frac12\right)}
	\int_0^\pi
	e^{n\|\vec\theta\|^2\cos\varphi}
	(\sin\varphi)^{n-2}
	d\varphi
	,
\]
where
$B(x,y)$
is the beta function.
\end{theorem}
\begin{proof}
See Sec. \ref{sec_proof_type2}.
\end{proof}

For a test to be SI,
it is necessary 
that the type II error probability 
does not depend on 
$\eta\in\Sqz^m$.
In the following theorem,
(i) implies that 
$T^\HH_\alpha$
is not SI.
Moreover,
(ii),implies that 
$T^\HH_\alpha$
is no better than 
$T^\trivial_\alpha$
in the minimax criterion.

\begin{theorem}\label{th_not_si}
(i)
If $\vec\theta\ne\vec0_m$,
then $\beta_{\hat\rho_{\vec\theta,\eta,N}^{\ox n}}[T^\HH_\alpha]$
depends on $\eta\in\Sqz^m$.
\\
(ii)
It holds that 
$
	\sup_{\eta\in\Sqz^m}
	\beta_{\hat\rho_{\vec\theta,\eta,N}^{\ox n}}[T^\HH_\alpha]
	=1-\alpha
$.
\end{theorem}
\begin{proof}
See Sec. \ref{sec_proof_not_si}.
\end{proof}

For the case of 
$m=1$,
$n=3$
and $N=0$,
the type II error probabilities of 
$T^\inv_\alpha$
and 
$T^\HH_\alpha$
are plotted in Figure.
The figure says that 
$\beta_{\hat\rho_{\theta,\eta,0}^{\ox3}}[T^\inv_\alpha]<
\beta_{\hat\rho_{\theta,\eta,0}^{\ox3}}[T^\HH_\alpha]$
if $\theta\fallingdotseq0$,
and that 
$\beta_{\hat\rho_{\theta,\eta,0}^{\ox3}}[T^\inv_\alpha]>
\beta_{\hat\rho_{\theta,\eta,0}^{\ox3}}[T^\HH_\alpha]$
if $\theta\gg0$.
If $\eta=O_2$
We can evaluate 
$\beta_{\hat\rho_{\theta,O_2,0}^{\ox3}}[T^\inv_\alpha]$
with 
$\beta_{\hat\rho_{\theta,O_2,0}^{\ox3}}[T^\HH_\alpha]$
as follows.

\begin{theorem}\label{th_comparison}
Consider the case of 
$m=1$,
$n=3$,
$\eta=O_2$,
$N=0$
and 
$0<\alpha<1$.
\\
(i)
There exists $s>0$
such that, 
if $0<|\theta|<s$,
then 
$\beta_{\hat\rho_{\theta,0}^{\ox 3}}[T^\inv_\alpha]
<\beta_{\hat\rho_{\theta,0}^{\ox 3}}[T^\HH_\alpha]$
holds.
\\
(ii) 
There exists $t>0$
such that, 
if $\theta>t$,
then 
$\beta_{\hat\rho_{\theta,0}^{\ox 3}}[T^\inv_\alpha]
>\beta_{\hat\rho_{\theta,0}^{\ox 3}}[T^\HH_\alpha]$
holds.
\end{theorem}
\begin{proof}
See Sec. \ref{sec_proof_comparison}.
\end{proof}

\def\thefigure{}
\begin{figure}
\centering
\includegraphics[clip,width=12cm]{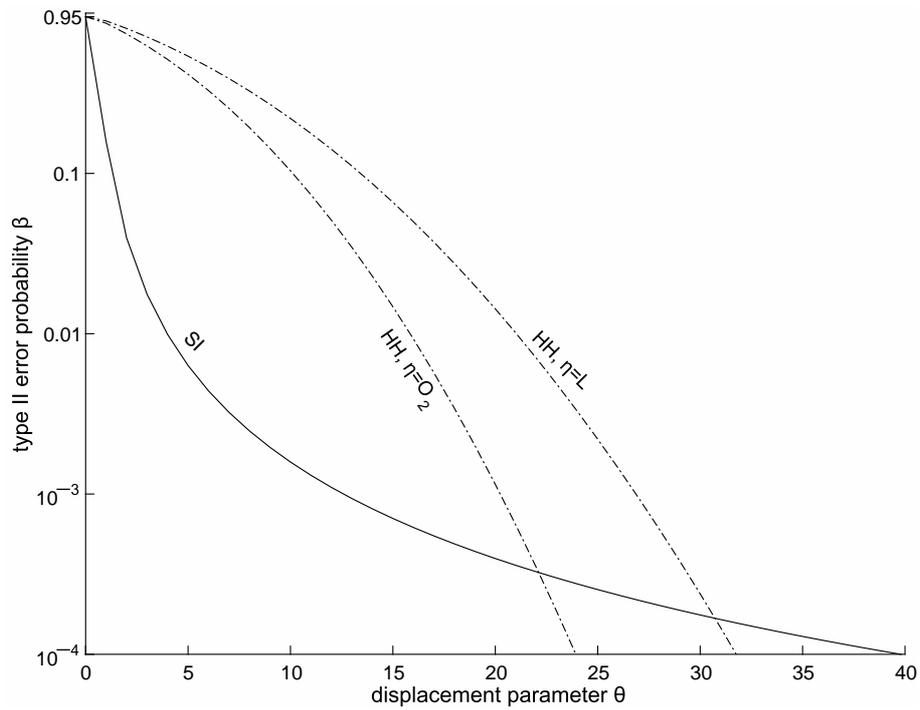}
\caption[Text excluding the matrix]{The type II error probabilities of 
$T^\inv_\alpha$ and $T^\HH_\alpha$
with $\alpha=0.95$ for $m=1$, $n=3$ and $N=0$.
The solid line is $\beta_{\hat\rho_{\theta,O_2,0}^{\ox3}}[T^\inv_\alpha]$.
The lower dashed line is $\beta_{\hat\rho_{\theta,O_2,0}^{\ox3}}[T^\HH_\alpha]$.
The upper dashed line is $\beta_{\hat\rho_{\theta,L,0}^{\ox3}}[T^\HH_\alpha]$,
where 
$L=
\begin{pmatrix}
0&1\\1&0
\end{pmatrix}$.
}
\end{figure}

\section{Proofs of the theorems}\label{sec_proofs}

\subsection{Matrix notations of operators}

Let 
$m$ be the mode size, 
and let 
$n$ be the sample size. 
For any $\mu,\nu\in\mathbb N$, 
let 
$\Mat_{m,n}^{\mu,\nu}$ 
be the set of $\mu$-by-$\nu$ matrices 
whose entries belong to 
$\cL(\cH^{\ox mn})$. 
Identifying 
$1$ with $\hat I^{\ox mn}$, 
we regard that 
$\Mat_\bC^{\mu,\nu}$ 
is a vector subspace of 
$\Mat_{m,n}^{\mu,\nu}$.
Let 
$\Mat_{m,n}^\mu$ 
be 
$\Mat_{m,n}^{\mu,\mu}$. 
We will use 
$L\in\cL(\cH^{\ox mn})$
which can be written by 
a linear combination of 
entries of 
$X\in\Mat_{m,n}^{\mu,\nu}$.

Define a linear map
$\tr_\mu:
\Mat_{m,n}^\mu
\to
\cL(\cH^{\ox mn})$
by
$\tr_\mu[X]
=\sum_{i=1}^\mu X_{i,i}$,
where
$X_{i,j}$ is the $(i,j)$-th entry of 
$X$.
\donotdisplay
{
On the other hand,
define
$\Tr:
\Mat_{m,n}^{\mu,\nu}
\to
\Mat_\bC^{\mu,\nu}$
as the entry-wise trace, 
that is, 
$(\Tr[X])_{i,j}
=\Tr[X_{i,j}]$. 
}
For 
$X\in\Mat_{m,n}^{\mu,\nu}$ 
and for
$L\in\cL(\cH^{\ox mn})$,
define
$XL,LX,[X,L]\in\Mat_{m,n}^{\mu,\nu}$
by
$(XL)_{i,j}=X_{i,j}L$,
$(LX)_{i,j}=LX_{i,j}$ 
and 
$[X,L]=XL-LX$, 
respectively. 
If $X\in\Mat_\bC^{\mu,\nu}$
and $L\in\cL(\cH^{\ox mn})$, 
then 
it holds that 
\begin{align}
	[X,L]
	=
	O_{\mu,\nu}
	, 
	\label{eq_zero_matrix}
\end{align}
where 
$O_{\mu,\nu}\in\Mat_\bC^{\mu,\nu}$
is 
the zero matrix. 
For
$X\in\Mat_{m,n}^{\lambda,\mu}$
and for
$Y\in\Mat_{m,n}^{\mu,\nu}$,
define the product 
$XY\in\Mat_{m,n}^{\lambda,\nu}$
by
$(XY)_{i,j}
=\sum_{k=1}^\mu X_{i,k}Y_{k,j}$.
For
$X\in\Mat_{m,n}^{\mu,\nu}$, 
define
$X^*\in\Mat_{m,n}^{\nu,\mu}$
by
$(X^*)_{i,j}
=(X_{j,i})^*$,
and 
define
${^t}X\in\Mat_{m,n}^{\nu,\mu}$
by
$({^t}X)_{i,j}=X_{j,i}$. 
The operations
$X\mapsto X^*$
and
$X\mapsto{^t}X$
are commutative 
as
$({^t}X)^*
={^t}(X^*)$,
and it holds that
$({^t}X^*)_{i,j}
=X_{i,j}^*$.
For $X\in\Mat_{m,n}^{\lambda,\mu}$,
for $Y\in\Mat_{m,n}^{\mu,\nu}$
and
for $C\in\mathrm{Mat}_\mathbb C^{\mu,\nu}$,
it holds that
$(XY)^*=Y^* X^*$
and that
${^t}(XC)=({^t}C)({^t}X)$.
Moreover,
for
$X\in\Mat_{m,n}^\mu$
and
for
$C\in\mathrm{Mat}_\mathbb C^\mu$,
it holds that
\begin{align}
	\tr_\mu[XC]
	=
	\sum_{i=1}^\mu
	\sum_{j=1}^\mu
	X_{i,j}C_{j,i}
	=
	\tr_\mu[CX]
	.
	\label{eq_tr}
\end{align}

For $j\in\{1,2,...,n\}$, 
define 
$\vec a_j
\in\Mat_{m,n}^{m,1}$
and 
$\dvec a_j
\in\Mat_{m,n}^{2m,1}$ 
by 
\[
	\vec a_j
	=
	\begin{pmatrix}
	\hat a_{1,j}\\\vdots\\\hat a_{m,j}
	\end{pmatrix}
	\mbox{ and }
	\dvec a_j
	=
	\begin{pmatrix}
	\vec a_j\\{^t}\vec a_j^*
	\end{pmatrix}
	,
\]
respectively.
Define 
$\tilde a_{m,n}
\in\Mat_{m,n}^{m,n}$
by 
$\tilde a_{m,n}
=(\vec a_1,\vec a_2,...,\vec a_n)$.
Define 
$\dtilde a_{m,n}
\in\Mat_{m,n}^{2m,n}$
by 
$\dtilde a_{m,n}
=(\dvec a_1,\dvec a_2,...,\dvec a_n)$.
For $\mu\in\bN$, 
let $I_\mu\in\Mat_\bR^\mu$ 
be the identity matrix, 
let $O_\mu\in\Mat_\bR^\mu$ 
be the zero matrix, 
and let 
\[
	K_\mu
	=
	\begin{pmatrix}
	I_\mu&O_\mu
	\\
	O_\mu&-I_\mu
	\end{pmatrix}
	\in\Mat_\bR^{2m}
	.
\]
If $n=1$,
then,
it holds that 
$\hat s_\eta
=2^{-1}\dvec a_1^*K_m\eta\dvec a_1
-\dtilde c_\eta\hat I^{\ox m}$, 
where 
$\dtilde c_\eta=4^{-1}\tr_{2m}[K_m\eta]
\in\sqrt{-1}\bR$. 
For any $m,n\in\bN$,
and 
for any $\eta\in\Sqz^m$,
define 
$\dhat u_\eta
\in\cL(\cH^{\ox mn})$
by 
\[
	\dhat u_\eta
	=
	\frac{1}{2}
	\tr_n[\dtilde a_{m,n}^*K_m\eta\dtilde a_{m,n}]
	-n\dtilde c_\eta\hat I^{\ox mn}
	.
\]
Then,
it holds that 
$
	\exp(\dhat u_\eta)=\hat S_\eta^{\ox n}
$.
For any 
$A\in\Ant_\bC^m$,
define 
$\hat u_A\in\cL(\cH^{\ox mn})$ 
by 
\[
	\hat u_A
	=
	\tr_n[\tilde a_{m,n}^*A\tilde a_{m,n}]
	.
\]
For $\eta\in\Sqz^m$,
if the anti-hermitian part 
is $A$,
and 
if the symmetric part 
is $O_m$,
then it holds that 
$\hat u_A=\dhat u_\eta$.

For $A\in\Ant_\bC^n$, 
define 
$\hat v_A\in\cL(\cH^{\ox mn})$ 
by 
\[
	\hat v_A
	=
	\tr_m[{^t}\tilde a_{m,n}^*A({^t}\tilde a_{m,n})]
	.
\]
For $j,k\in\{1,2,...,n\}$
with $j<k$, 
define $J_{j,k}\in\Ant_\bR^n$ 
by 
\begin{align}
	J_{j,k}
	=
	\begin{pmatrix}
	O_{j-1}\\
	& 0 & & -1\\
	&   & O_{k-j-1}\\
	& 1 & & 0\\
	& & & & O_{n-k}
	\end{pmatrix}
	.
	\label{eq_j_j_k}
\end{align}
If $j=k$,
let $J_{j,k}=O_n$.
If $j>k$,
let $J_{j,k}=-J_{k,j}$.
It holds that 
$\hat v_{J_{j,k}}
=\hat v_{j,k}$. 

By Eq. (\ref{eq_hat_a_ket_theta}),
for any 
$A\in\Ant_\bC^m$ 
and for any 
$B\in\Ant_\bC^n$,
it holds that 
\begin{align}
	\hat u_A\ket0^{\ox mn}
	=
	\hat v_B\ket0^{\ox mn}
	=
	0
	\in\cH^{\ox mn}
	.
	\label{eq_ket_zero}
\end{align}

\subsection{Proof of Theorem \ref{th_inv} ($\hat v_{j,k}$ is SI)}\label{sec_proof_inv}

We first consider 
how $\hat s_\eta$ is represented 
in the case of $n=1$.
To simplify the notation, 
$\dvec a_1\in\Mat_{m,1}^{2m,1}$
is denoted by 
$\dvec a={^t}(\hat a_1,\hat a_2,...,\hat a_m,\hat a_1^*,\hat a_2^*,...,\hat a_m^*)$.

\begin{lemma}\label{lem_s_rep}
If $n=1$,
then, 
for any $m\in\bN$, 
and 
for any $\eta\in\mathrm{Sqz}^m$, 
it holds that
$[\dvec a,\dhat s_\eta]
=\eta\dvec a$.
\end{lemma}
\begin{proof}
Let $A\in\Ant_\bC^m$
be the anti-hermitian part 
of $\eta$,
and let
$S\in\Sym_\bC^m$
be the symmetric part 
of $\eta$. 
For $i\in\{1,2,...,m\}$,
let
\[
	\hat F_i
	=
	\Big(
	\sum_{j=1}^m
	A_{i,j}\hat a_i^*\hat a_j
	\Big)
	+\frac{S_{i,i}}{2}(\hat a_i^*)^2
	+\Big(
	\sum_{j\ne i}
	S_{i,j}\hat a_i^*\hat a_j^*
	\Big)
	\in\cL(\cH^{\ox m})
	.
\]
Then, 
for any $i\le m$, 
it holds that
\begin{align}
	[\hat a_i,\hat s_\eta]
	=&
	[\hat a_i,\hat F_i]
	=
	\Big(
	\sum_{j=1}^m
	A_{i,j}\hat a_j
	\Big)
	+S_{i,i}\hat a_i^*
	+\Big(
	\sum_{j\ne i}
	S_{i,j}\hat a_j^*
	\Big)
	=
	\eta_i\dvec a
	,
	\label{eq_lem_1_1}
\end{align}
where 
$\eta_i\in\Mat_\bC^{1,2m}$ 
is the $i$-th row vector of $\eta$.
Similarly,
let
\[
	\hat G_i
	=
	\Big(
	\sum_{j=1}^m
	A_{j,i}\hat a_j^*\hat a_i
	\Big)
	-\frac{\bar S_{i,i}}{2}\hat a_i^2
	-\Big(
	\sum_{j\ne i}
	\bar S_{i,j}\hat a_i\hat a_j
	\Big)
	\in\cL(\cH^{\ox m})
	.
\]
This definition is equivalent to 
\[
	\hat G_i
	=
	-\frac{\bar S_{i,i}}{2}\hat a_i^2
	-\Big(
	\sum_{j\ne i}
	\bar S_{i,j}\hat a_i\hat a_j
	\Big)
	-\Big(
	\sum_{j=1}^m
	\bar A_{i,j}\hat a_j^*\hat a_i
	\Big)
	.
\]
For any $i\le m$, 
it holds that
\begin{align}
	[\hat a_i^*,\hat s_\eta]
	=&
	[\hat a_i^*,\hat G_i]
	=
	\bar S_{i,i}\hat a_i
	+\Big(
	\sum_{j\ne i}
	\bar S_{i,j}\hat a_j
	\Big)
	+\Big(
	\sum_{j=1}^m
	\bar A_{i,j}\hat a_j^*
	\Big)
	=
	\eta_{i+m}\dvec a
	.
	\label{eq_lem_1_2}
\end{align}
By Eqs. 
(\ref{eq_lem_1_1}) 
and 
(\ref{eq_lem_1_2}), 
we obtain 
$[\dvec a,\hat s_\eta]
=\eta\dvec a$. 
\end{proof}

Then,
for any $n\ge1$,
we obtain a representation of 
$\dhat u_\eta$ 
on $\dtilde a_{m,n}$.

\begin{lemma}\label{lem:squeezing_representation}
For any $m,n\in\bN$, 
and 
for any $\eta\in\mathrm{Sqz}^m$, 
it holds that
$[\dtilde a_{m,n},\dhat u_\eta]
=\eta\dtilde a_{m,n}$.
\end{lemma}
\begin{proof}
For any $i\in\{1,2,...,m\}$,
and 
for any $j,k\in\{1,2,...,n\}$
with 
$j\ne k$,
it holds that 
$[\hat a_{i,j},\dvec a_k]
=[\hat a_{i,j}^*,\dvec a_k]
=O_{2m,1}$.
Hence, 
by Lemma 
\ref{lem_s_rep},
we obtain 
$[\dtilde a_{m,n},\dhat u_\eta]
=\eta\dtilde a_{m,n}$.
\end{proof}

\begin{lemma}\label{lem_u_A}
For any $A\in\Ant_\bC^m$,
it holds that
$[\tilde a_{m,n},\hat u_A]=A\tilde a_{m,n}$
and that 
$[{^t}\tilde a_{m,n}^*,\hat u_A]=\bar A({^t}\tilde a_{m,n}^*)$.
\end{lemma}
\begin{proof}
Applying 
Lemma \ref{lem:squeezing_representation} 
to the case where 
the symmetric part of $\eta$ is $O_m$,
we obtain 
the statement.
\end{proof}

\begin{lemma}\label{lem:c_right_rep}
For any $A\in\Ant_\bC^n$,
it holds that
$[\tilde a_{m,n},\hat v_A]=-\tilde a_{m,n}\bar A$
and that 
$[{^t}\tilde a_{m,n}^*,\hat v_A]=-{^t}\tilde a_{m,n}^*A$.
\end{lemma}
\begin{proof}
By Lemma \ref{lem_u_A},
it holds that 
$[{^t}\tilde a_{m,n},\hat v_A]=A({^t}\tilde a_{m,n})$
and that 
$[\tilde a_{m,n}^*,\hat v_A]=\bar A\tilde a_{m,n}^*$.
By transposing, 
we obtain 
the statement.
\end{proof}

\begin{lemma}\label{lem:r_right_rep}
For any $m,n\in\bN$, 
and for any 
$B\in\Ant_\bR^n$,
it holds that 
$[\dtilde a_{m,n},\hat v_B]
=-\dtilde a_{m,n}B$.
\end{lemma}
\begin{proof}
By
Lemma \ref{lem:c_right_rep}, 
and by $\bar B=B$, 
we have
$[\dtilde a_{m,n},\hat v_B]=-\dtilde a_{m,n}B$.
\end{proof}

\begin{lemma}\label{lem_inv}
For any $\eta\in\mathrm{Sqz}^m$
and
for any $B\in\mathrm{Ant}_\mathbb R^n$,
it holds that
$[\dhat u_\eta,\hat v_B]=0$.
\end{lemma}
\begin{proof}
By 
Eq. (\ref{eq_zero_matrix}),
it holds that
$[K_m\eta,\hat v_B]=O_{2m}$. 
Hence,
we have
\begin{align*}
	2[\dhat u_\eta,\hat v_B]
	=&
	\tr_n(\dtilde a_{m,n}^*K_m\eta\dtilde a_{m,n})\hat v_B
	-
	\tr_n(\dtilde a_{m,n}^*[K_m\eta,\hat v_B]\dtilde a_{m,n})
	\\
	&-
	\hat v_B\tr_n(\dtilde a_{m,n}^*K_m\eta\dtilde a_{m,n})
	\\
	=&
	\tr_n(\dtilde a_{m,n}^*K_m\eta[\dtilde a_{m,n},\hat v_B])
	+
	\tr_n([\dtilde a_{m,n}^*,\hat v_B]K_m\eta\dtilde a_{m,n})
	.
\end{align*}
By Lemma \ref{lem:r_right_rep},
we have 
\begin{align*}
	2[\dhat u_\eta,\hat v_B]
	=&
	-\tr_n(\dtilde a_{m,n}^*K_m\eta\dtilde a_{m,n}B)
	+
	\tr_n(B\dtilde a_{m,n}^*K_m\eta\dtilde a_{m,n})
	.
\end{align*}
By Eq. (\ref{eq_tr}), 
the right-hand side is zero.
\end{proof}

\begin{proof}[{\bf Proof of Theorem \ref{th_inv}}]
Using $J_{j,k}\in\Ant_\bR^n$
of (\ref{eq_j_j_k}),
it holds that 
$\hat v_{j,k}=\hat v_{J_{j,k}}$.
Hence,
by Lemma \ref{lem_inv}, 
it holds that 
$[\dhat u_\eta,\hat v_{j,k}]=0$.
As 
$\hat S_\eta^{\ox n}=\exp(\dhat u_\eta)$,
we obtain 
$(\hat S_\eta^{\ox n})^*\hat v_{j,k}\hat S_\eta^{\ox n}
=\hat v_{j,k}$.
\end{proof}

\subsection{Other properties of $\hat u_A$ and $\hat v_B$}

We will use 
the following lemma in Secs. 
\ref{sec_proof_radial}, 
\ref{sec_n2}
and 
\ref{sec_proof_type2}.

\begin{lemma}\label{lem:v_a_b_c}
For any $m,n\in\bN$, 
and for any $A,B\in\Ant_\bC^n$, 
it holds that
$[\hat v_A,\hat v_B]
=\hat v_{[A,B]}$.
\end{lemma}
\begin{proof}
By 
Eq. (\ref{eq_zero_matrix}),
it holds that
$[A,\hat v_B]=O_n$.
Hence,
we have
\begin{align*}
	[\hat v_A,\hat v_B]
	=&
	\tr_m({^t}\tilde a_{m,n}^*A\,{^t}\tilde a_{m,n})\hat v_B
	-
	\tr_m({^t}\tilde a_{m,n}^*[A,\hat v_B]\,{^t}\tilde a_{m,n})
	\\
	&
	-
	\hat v_B\tr_m({^t}\tilde a_{m,n}^*A\,{^t}\tilde a_{m,n})
	\\
	=&
	\tr_m({^t}\tilde a_{m,n}^*A[{^t}\tilde a_{m,n},\hat v_B])
	+
	\tr_m([{^t}\tilde a_{m,n}^*,\hat v_B]A\,{^t}\tilde a_{m,n})
	.
\end{align*}
Moreover, 
by Lemma \ref{lem:c_right_rep},
we have 
\begin{align*}
	[\hat v_A,\hat v_B]
	=&
	\tr_m({^t}\tilde a_{m,n}^*AB\ {^t}\tilde a_{m,n})
	-
	\tr_m({^t}\tilde a_{m,n}^*BA\ {^t}\tilde a_{m,n})
	.
\end{align*}
Hence,
we have
$[\hat v_A,\hat v_B]
=\hat v_{[A,B]}$.
\end{proof}

We will use 
the following lemma in 
Sec. \ref{sec_rep}.

\begin{lemma}\label{lem:commute_s_r}\label{lem_u_a_v_b_commute}
For any $A\in\Ant_\bC^m$
and
for any $B\in\Ant_\bC^n$,
it holds that 
$[\hat u_A,\hat v_B]=0$.
\end{lemma}
\begin{proof}
By 
Eq. (\ref{eq_zero_matrix}),
it holds that 
$[A,\hat v_B]=O_m$. 
Hence,
we have
\begin{align*}
	[\hat u_A,\hat v_B]
	=&
	\tr_n(\tilde a_{m,n}^*A\tilde a_{m,n})\hat v_B
	-
	\tr_n(\tilde a_{m,n}^*[A,\hat v_B]\tilde a_{m,n})
	\\
	&-
	\hat v_B\tr_n(\tilde a_{m,n}^*A\tilde a_{m,n})
	\\
	=&
	\tr_n(\tilde a_{m,n}^*A[\tilde a_{m,n},\hat v_B])
	+
	\tr_n([\tilde a_{m,n}^*,\hat v_B]A\tilde a_{m,n})
	.
\end{align*}
By Lemma 
\ref{lem:c_right_rep},
we have 
\begin{align*}
	[\hat u_A,\hat v_B]
	=&
	-\tr_n(\tilde a_{m,n}^*A\tilde a_{m,n}\bar B)
	+
	\tr_n(\bar B\tilde a_{m,n}^*A\tilde a_{m,n})
	.
\end{align*}
By Eq. (\ref{eq_tr}), 
the right-hand side is zero.
\end{proof}

\subsection{Proof of Theorem \ref{th_radial} (MwRI functions)}\label{sec_proof_radial}

\donotdisplay
{
A vector $f\in\cH^{\ox mn}$
is regarded as a function 
$f:\Mat_\bR^{m,n}\to\bC$,
where 
the $(i,j)$-th entry 
$x_{i,j}$ 
of the matrix variable 
$X\in\Mat_\bR^{m,n}$
is specified by 
$\hat a_{i,j}f
=(x_{i,j}f+\partial f/\partial x_{i,j})/\sqrt2$.
Then, 
$\hat v_{j,k}\in\cA(\cH^{\ox mn})$
is represented by 
$J_{j,k}\in\Ant_\bR^n$
defined in 
Eq. (\ref{eq_j_j_k}).
}

We consider how 
$\hat v_{j,k}\in\cL(\cH^{\ox mn})$
is represented by 
$f(X)\in\cH^{\ox mn}$
where $X\in\Mat_\bR^{m,n}$.

\begin{lemma}\label{lem_v_j_k_f_x}
For any $r\in\bR$,
it holds that 
$
	e^{r\hat v_{j,k}}f(X)
	=
	f(Xe^{rJ_{j,k}})
$.
\end{lemma}
\begin{proof}
Let $x_{i,j}$ 
be the $(i,j)$-th entry 
of $X$.
It holds that 
$\hat v_{j,k}
=\sum_{i=1}^m
(\hat a_{i,k}^*\hat a_{i,j}-\hat a_{i,j}^*\hat a_{i,k})
=\sum_{i=1}^m
(x_{i,k}\partial/\partial x_{i,j}-x_{i,j}\partial/\partial x_{i,k})$.
Let
$(s_{i,j,k},\varphi_{i,j,k})$
be a polar coordinate of
$(x_{i,j},x_{i,k})$
defined by
$x_{i,j}=s_{i,j,k}\cos\varphi_{i,j,k}$
and
$x_{i,k}=s_{i,j,k}\sin\varphi_{i,j,k}$.
Then,
it holds that
\[
	-\frac{\partial}{\partial\varphi_{i,j,k}}
	=
	-\frac{\partial x_{i,j}}{\partial\varphi_{i,j,k}}
	\frac{\partial}{\partial x_{i,j}}
	-\frac{\partial x_{i,k}}{\partial\varphi_{i,j,k}}
	\frac{\partial}{\partial x_{i,k}}
	=
	x_{i,k}\frac{\partial}{\partial x_{i,j}}
	-
	x_{i,j}\frac{\partial}{\partial x_{i,k}}
	,
\]
and so that 
$\hat v_{j,k}
=-\sum_{i=1}^m\partial/\partial\varphi_{i,j,k}$. 
As 
$
	e^{-t\partial/\partial\varphi}
	=
	\sum_{k=0}^\infty
	(-t\partial/\partial\varphi)^k/k!
$ holds,
it holds that 
$
	e^{-t\partial/\partial\varphi}
	\varphi^n
	=
	\sum_{k=0}^n
	\begin{pmatrix}n\\k\end{pmatrix}
	\varphi^{n-k}(-t)^k
	=
	(\varphi-t)^n
$.
Hence,
we obtain 
$\exp(r\hat v_{j,k})f(X)
=f({^t}[\exp(-rJ_{j,k})({^t}X)])
=f\big(X\exp(rJ_{j,k})\big)$.
\end{proof}

For $\hat L\in\cL(\cH^{\ox mn})$,
let 
$\Null(\hat L)
\subset\cH^{\ox mn}$
be the nullspace 
$\{f\in\cH^{\ox mn}\mid\hat Lf=0\}$.
The following lemma shows that 
$
	\bigcap_{k=1}^{n-1}
	\Null(\hat v_{k,n})
$
is the set of 
square-integrable MwRI functions.

\begin{lemma}\label{lem_nullspace}
It holds that 
$
	\bigcap_{k=1}^{n-1}
	\Null(\hat v_{k,n})
	=
	\bigcap_{A\in\Ant_\bR^n}
	\Null(\hat v_A)
$.
\end{lemma}
\begin{proof}
By Lemma \ref{lem:v_a_b_c},
it holds that 
$[\hat v_{i,j},\hat v_{j,k}]=\hat v_{k,i}$
for any
$i,j,k\in\{1,2,...,n\}$.
Hence, 
it holds that 
\[
	\Null(\hat v_{i,j})\cap\Null(\hat v_{j,k})
	=
	\Null(\hat v_{i,j})\cap\Null(\hat v_{j,k})\cap\Null(\hat v_{k,i})
	.
\]
Hence,
we have 
$
	\bigcap_{k=1}^{n-1}
	\Null(\hat v_{k,n})
	=
	\bigcap_{i=1}^n
	\allowbreak
	\bigcap_{j=1}^n
	\Null(\hat v_{i,j})
$.
Since 
$J_{j,k}\in\Ant_\bR^n$
holds,
$
	\bigcap_{i=1}^n
	\bigcap_{j=1}^n
	\Null(\hat v_{i,j})
	\supset
	\bigcap_{A\in\Ant_\bR^n}
	\Null(\hat v_A)
$ holds.
Since any $\hat v_A$
is a linear combination of 
$\hat v_{j,k}$,
we have 
$
	\bigcap_{i=1}^n
	\bigcap_{j=1}^n
	\Null(\hat v_{i,j})
	\subset
	\bigcap_{A\in\Ant_\bR^n}
	\Null(\hat v_A)
$.
Hence,
we obtain the statement.
\end{proof}

\donotdisplay
{
\begin{lemma}\label{lem_radial}
Let 
$\cK=\bigcap_{k=1}^{n-1}\Null(\hat v_{k,n})$.
If $m=1$,
then 
$\cK=\cR$.
\end{lemma}
\begin{proof}
By Lemma \ref{lem_nullspace},
it holds that 
$\cK=\bigcap_{k=1}^{n-1}\Null(\hat v_{k,k+1})$. 
Let 
$(r,s_1,s_2,...,s_{n-1})$
be a polar coordinate 
given by 
$r
=\sqrt{{^t}\vec x\vec x}$
and 
\[
	\vec x
	=
	r
	\begin{pmatrix}
	\cos s_1
	\\
	\cos s_2\sin s_1
	\\
	\vdots
	\\
	\cos s_{n-1}\prod_{k=1}^{n-2}\sin s_k
	\\
	\prod_{k=1}^{n-1}\sin s_k
	\end{pmatrix}
	.
\]
By Lemma 
\ref{lem_v_j_k_f_x}, 
$f
\in\Null(\hat v_{k,k+1})$ 
if and only if 
$f$ does not depend on 
$s_k$. 
Hence, 
$f\in\cK$ 
if and only if 
$f\in\cR$.
\end{proof}
}

\begin{proof}[{\bf Proof of Theorem \ref{th_radial}}]
Let 
$\hat T=\sum_{k=1}^{n-1}\hat v_{k,n}\hat v_{k,n}^*$. 
Then,
it holds that 
$
	\Null(\hat T)
	=
	\bigcap_{k=1}^{n-1}
	\Null(\hat v_{k,n}\hat v_{k,n}^*)
$.
Since 
$\Null(\hat v_{k,n}\hat v_{k,n}^*)
=\Null(\hat v_{k,n}^*\hat v_{k,n})
=\Null(\hat v_{k,n})$
holds,
we have 
$
	\bigcap_{k=1}^{n-1}
	\Null(\hat v_{k,n}\hat v_{k,n}^*)
	=
	\bigcap_{k=1}^{n-1}
	\Null(\hat v_{k,n})
$.
By Lemma 
\ref{lem_nullspace},
$\Null(\hat T)$
is the set of 
square-integrable MwRI functions.
Since 
$\hat T_\inv=\hat R^*\hat T\hat R$
holds,
we have 
$\cK_0
=\Null(\hat T_\inv)
=\hat R^*\Null(\hat T)$.
By 
$\hat R
=\hat R_{n-1}\hat R_{n-2}\cdots\hat R_1$,
and $R_k
=\exp(
	\arctan(\sqrt k)
	\hat v_{k,k+1}
)$,
we have 
$\hat R^*\Null(\hat T)
=\Null(\hat T)$. 
\end{proof}

\subsection{Actions of $\hat u_A$ and $\hat v_B$ on $\vec\theta$}\label{sec_rep}

For $\theta\in\mathbb C$, 
the displacement operator 
$\hat D_\theta\in\cU(\cH)$ 
is defined by 
$\hat D_\theta
=\exp(\hat d_\theta)$,
where
$\hat d_\theta
=\theta\hat a^*-\bar\theta\hat a$.
It holds that
\begin{align}
	\hat D_\theta^*\hat a\hat D_\theta
	=\hat a+\theta\hat I
	\mbox{ and }
	\hat D_\theta^*\hat a^*\hat D_\theta
	=\hat a^*+\bar\theta\hat I
	\label{eq_displace}
\end{align}
because of 
$[\hat a,\hat d_\theta]=\theta\hat I$,
and 
$[\hat a^*,\hat d_\theta]=\bar\theta\hat I$,
respectively.
Using Taylor expansion,
for $r\in\mathbb R$,
we have 
$\hat D_rf(x)
=\exp(-ir\sqrt2\hat p)f(x)
=f(x-\sqrt2r)$.
By Baker-Hausdorff formula
\cite{leonhardt},
it holds that
$\hat D_{r+is}
=e^{-i r s}e^{i\sqrt2s\hat q}e^{-i\sqrt2r\hat p}$
for $r,s\in\mathbb R$.
Hence, 
we have
\begin{align}
	\hat D_{r+is}\ket0(x)
	=
	\frac{1}{\pi^{1/4}}
	\exp\Big(
	-irs
	-\frac{(x-\sqrt2r)^2}{2}
	+i\sqrt{2}sx
	\Big)
	=
	\ket{r+is}(x)
	.
	\label{eq_d_ris}
\end{align}
For 
$\vec z
={^t}(z_1,z_2,...,z_m)
\in\Mat_\bC^{m,1}$,
let 
$\Ket{\vec z}
=\ket{z_1}\ox\ket{z_2}\ox\cdots\ox\ket{z_m}
\in\cH^{\ox m}$.
For 
$Z
=(\vec Z_1,\vec Z_2,...,\vec Z_n)
\in\Mat_\bC^{m,n}$,
let 
$\ket Z=\ket{\vec Z_1}\ox\ket{\vec Z_2}\ox\cdots\ox\ket{\vec Z_n}\in\cH^{\ox mn}$.
For $W\in\Mat_\bC^{m,n}$,
let 
\[
	\hat d_W
	=
	\tr_m[W\tilde a_{m,n}^*]
	-\tr_m[\tilde a_{m,n}W^*]
	,
\]
and 
define 
a displacement operator 
$\hat D_W\in\cU(\cH^{\ox mn})$
by 
$
	\hat D_W=\exp(\hat d_W)
$.
By Eq. 
(\ref{eq_d_ris}),
it holds that 
$\hat D_W\Ket{O_{m,n}}=\ket W$,
where 
$O_{m,n}\in\Mat_\bC^{m,n}$
is the zero matrix.
For 
$A\in\Ant_\bC^m$,
let 
$\hat U_A
=\exp(\hat u_A)
\in\cU(\cH^{\ox mn})$.
For $B\in\Ant_\bC^n$, 
let 
$\hat V_B
=\exp(\hat v_B)
\in\cU(\cH^{\ox mn})$.

\begin{lemma}\label{lem_a_b_z}
For any $A\in\Ant_\bC^m$, 
for any $B\in\Ant_\bC^n$ 
and 
for any $Z\in\Mat_\bC^{m,n}$, 
it holds that 
\[
	\hat U_A\hat V_B\ket Z
	=
	\hat V_B\hat U_A\ket Z
	=
	\Ket{e^AZe^{-\bar B}}
	.
\]
\end{lemma}
\begin{proof}
By Lemma
\ref{lem_u_a_v_b_commute},
it holds that 
$\hat U_A\hat V_B=\hat V_B\hat U_A$.
Hence, 
we just need to prove 
$
	\hat U_A\hat V_B\ket Z
	=
	\Ket{e^AZe^{-\bar B}}
$.
By Lemmas 
\ref{lem_u_A}
and 
\ref{lem:c_right_rep},
it holds that 
$\hat U_A^*
	\hat V_B^*
	\hat d_W
	\hat V_B
	\hat U_A
	=
	\hat d_{e^{-A}We^{\bar B}}
$.
By Eq (\ref{eq_ket_zero}),
it holds that 
$\hat U_A\hat V_B\Ket{O_{m,n}}
=\Ket{O_{m,n}}$.
Hence, 
\[
	\hat U_A\hat V_B\ket Z
	=
	\hat U_A\hat V_B\hat D_Z\hat V_B^*\hat U_A^*
	\hat U_A\hat V_B\Ket{O_{m,n}}
	=
	\hat D_{e^AZe^{-\bar B}}\Ket{O_{m,n}}
	=
	\Ket{e^AZe^{-\bar B}}
\]
is obtained.
\end{proof}

For 
$Z
=(\vec Z_1,\vec Z_2,...,\vec Z_n)
\in\Mat_\bC^{m,n}$,
and 
for $N\ge0$,
let 
$\hat\rho_{Z,N}
\in\cS(\cH^{\ox mn})$
be 
$
	\hat\rho_{\vec Z_1,N}\ox
	\hat\rho_{\vec Z_2,N}\ox
	\cdots\ox
	\hat\rho_{\vec Z_n,N}
$.

\begin{lemma}\label{lem_g_a_b_z}
For any $A\in\Ant_\bC^m$,
for any $B\in\Ant_\bC^n$,
for any $Z\in\Mat_\bC^{m,n}$
and 
for any $N\ge0$,
it holds that 
\[
	\hat U_A\hat V_B\hat\rho_{Z,N}\hat V_B^*\hat U_A^*
	=
	\hat V_B\hat U_A\hat\rho_{Z,N}\hat U_A^*\hat V_B^*
	=
	\hat\rho_{e^AZe^{-\bar B},N}
	.
\]
\end{lemma}
\begin{proof}
Let 
$
	\hat\sigma
	=
	\hat U_A\hat V_B
	\hat\rho_{Z,N}
	\hat V_B^*\hat U_A^*
$.
In the case of 
$N=0$,
by Lemma \ref{lem_a_b_z},
we obtain 
$\hat\sigma
=\hat\rho_{e^AZe^{-\bar B},0}$.
Consider 
the case of 
$N>0$.
By Lemma
\ref{lem_u_a_v_b_commute},
it holds that 
$
	\hat\sigma
	=
	\hat V_B\hat U_A
	\hat\rho_{Z,N}
	\hat U_A^*\hat V_B^*
$.
By Lemma \ref{lem_a_b_z},
it holds that 
\begin{align*}
	&
	\hat\sigma
	=
	\left.\overbrace{\int\cdots\int}^{2mn}\right._{\bR^{2mn}}
	\frac
	{e^{-\|W-Z\|^2/N}}
	{(\pi N)^{mn}}
	\Ket{e^AWe^{-\bar B}}
	\Bra{e^AWe^{-\bar B}}
	\prod_{i=1}^m
	\prod_{j=1}^n
	du_{i,j}
	dv_{i,j}
	,
\end{align*}
where
$\|W-Z\|^2
=\tr_m[(W-Z)(W^*-Z^*)]$, 
and where 
$u_{i,j}$
and 
$v_{i,j}$
are 
the $(i,j)$-th entries of  
$\re(W)$ and $\im(W)$,
respectively. 
Replace 
$\re(e^AWe^{-\bar B})$ 
by 
$R\in\Mat_\bR^{m,n}$,
and replace 
$\im(e^AWe^{-\bar B})$ 
by 
$S\in\Mat_\bR^{m,n}$.
Let 
$\vec u_j,\vec v_j,\vec r_j$ and $\vec s_j
\in\Mat_\bR^{m,1}$
be 
$j$-th column vectors of 
$\re(W)$,
$\im(W)$,
$R$ and $S$,
respectively.
Define 
$\vec u,\vec v,\vec r$ and $\vec s
\in\Mat_\bR^{mn,1}$
by
\[
	\vec u
	=
	\begin{pmatrix}
	\vec u_1\\\vdots\\\vec u_n
	\end{pmatrix}
	,\
	\vec v
	=
	\begin{pmatrix}
	\vec v_1\\\vdots\\\vec v_n
	\end{pmatrix}
	,\
	\vec r
	=
	\begin{pmatrix}
	\vec r_1\\\vdots\\\vec r_n
	\end{pmatrix}
	\mbox{ and }
	\vec s
	=
	\begin{pmatrix}
	\vec s_1\\\vdots\\\vec s_n
	\end{pmatrix}
	,
\]
respectively.
Let 
$r_{i,j}$ and $s_{i,j}$ 
be 
the $(i,j)$-th entries of 
$R$ and $S$,
respectively. 
Then, 
it holds that 
$\vec r+\sqrt{-1}\vec s
=C(\vec u+\sqrt{-1}\vec v)$,
where 
$C\in\Mat_\bC^{mn}$
is a unitary matrix 
given by 
the Kronecker product 
of $e^B$ and $e^A$ 
as 
$C=e^B\ox e^A$.
Let 
$D=B\ox I_m+I_n\ox A
\in\Ant_\bC^{mn}$.
Then, 
it holds that 
$C=e^D$. 
Define 
$\vec w,\vec t\in\Mat_\bR^{2mn,1}$
and 
$E\in\Ant_\bR^{2mn}$ by
\[
	\vec w
	=
	\begin{pmatrix}
	\vec u\\\vec v
	\end{pmatrix}
	,\
	\vec t
	=
	\begin{pmatrix}
	\vec r\\\vec s
	\end{pmatrix}
	\mbox{ and }
	E
	=
	\begin{pmatrix}
	\re(D)&-\im(D)\\
	\im(D)&\re(D)
	\end{pmatrix}
	,
\]
respectively.
Then, 
it holds that 
$\vec t=e^E\vec w$
and that 
$\det[e^E]=\exp(\tr_{2mn}[E])=1$.
Hence, 
the Jacobian of the replacement 
of 
$e^E\vec w$
by 
$\vec t$
is one.
Hence, 
we have 
\[
	\hat\sigma
	=
	\left.\overbrace{\int\cdots\int}^{2mn}\right._{\bR^{2mn}}
	\frac
	{e^{-\|e^{-A}Te^{\bar B}-Z\|^2/N}}
	{(\pi N)^{mn}}
	\ket T\bra T
	\prod_{i=1}^m
	\prod_{j=1}^n
	dr_{i,j}
	ds_{i,j}
	,
\]
where 
$T=e^AWe^{-\bar B}$.
Since
\begin{align*}
	&
	\|e^{-A}Te^{\bar B}-Z\|^2
	=
	\tr_m[(e^{-A}Te^{\bar B}-Z)(e^{-A}Te^{\bar B}-Z)^*]
	\\
	&=
	\tr_m[(T-e^AZe^{-\bar B})(T-e^AZe^{-\bar B})^*]
	=
	\|T-e^AZe^{-\bar B}\|^2
\end{align*}
holds,
we obtain 
$\sigma
=\hat\rho_{e^AZe^{-\bar B},N}$.
\end{proof}

\begin{lemma}\label{lem:dariano}
For any $\vec\theta\in\Mat_\bC^{m,1}$, 
it holds that
$\hat R
\hat\rho_{\vec\theta,N}^{\ox n}\hat R^*
=\hat\rho_{\vec0_m,N}^{\ox(n-1)}
\ox\hat\rho_{\sqrt n\vec\theta,N}$.
\end{lemma}
\begin{proof}
Let $\vec1_n\in\Mat_\bR^{n,1}$ 
be a vector whose entries are all one,
and let 
$\vec e_n\in\Mat_\bR^{n,1}$
be a unit vector whose $n$-th entry is one.
It holds that 
$\hat\rho_{\vec\theta,N}^{\ox n}=\hat\rho_{\vec\theta({^t}\vec1_n),N}$
and that 
$\hat\rho_{\vec0_m,N}^{\ox(n-1)}\ox\hat\rho_{\sqrt n\vec\theta,N}
=\hat\rho_{\sqrt n\vec\theta({^t}\vec e_n),N}$. 
Let 
$R_k
=\exp[(\arctan\sqrt k)J_{k,k+1}]
\in\Mat_\bR^n$,
and 
let 
$R=R_{n-1}R_{n-2}\cdots R_1
\in\Mat_\bR^n$.
By Lemma \ref{lem_g_a_b_z},
it holds that
$\hat R
\hat\rho_{\vec\theta,N}^{\ox n}
\hat R^*
=\hat\rho_{\vec\theta[{^t}(R\vec1_n)],N}$. 
If $t\in[0,\pi/2)$,
then
$\cos t=1/\sqrt{1+\tan^2t}$
and
$\sin t=\tan t/\sqrt{1+\tan^2t}$
hold.
Hence,
$\cos(\arctan\sqrt k)=1/\sqrt{k+1}$
and
$\sin(\arctan\sqrt k)=\sqrt k/\sqrt{k+1}$
hold.
Hence,
we have 
\[
	R_k
	=
	\frac{1}{\sqrt{k+1}}
	\begin{pmatrix}
	I_{k-1}\\
	& 1 & -\sqrt k \\
	& \sqrt k & 1 \\
	& & & I_{n-k-1}
	\end{pmatrix}
	.
\]
The $k$-th and $k+1$-th row vectors
of $R_kR_{k-1}\cdots R_1$
are recursively obtained 
by calculating 
the $2$-by-$2$ submatrix as 
\[
	\begin{pmatrix}
	\frac{1}{\sqrt{k+1}} & -\frac{\sqrt k}{\sqrt{k+1}}\\
	\frac{\sqrt k}{\sqrt{k+1}} & \frac{1}{\sqrt{k+1}}
	\end{pmatrix}
	\begin{pmatrix}
	\frac{1}{\sqrt k} & 0 \\
	0 & 1 
	\end{pmatrix}
	=
	\begin{pmatrix}
	\frac{1}{\sqrt k\sqrt{k+1}} & -\frac{\sqrt k}{\sqrt{k+1}} \\
	\frac{1}{\sqrt{k+1}} & \frac{1}{\sqrt{k+1}} 
	\end{pmatrix}
	.
\]
Hence,
the $n$-th row vector 
of $R$ 
is 
${^t}(1/\sqrt n\vec1_n)$. 
Since $R$ is an orthogonal matrix,
it holds that
$R\vec1_n=\sqrt n\vec e_n$.
\end{proof}

\subsection{Some properties related to the Fock vectors}

Let 
$\hat N\in\cL(\cH)$ 
be 
the number operator 
$\hat a^*\hat a$.
It holds that 
$[\hat a,\hat N]=\hat a$,
and so that 
$e^{-ir\hat N}\hat ae^{ir\hat N}=e^{ir}\hat a$
(${^\forall}r\in\bR$), 
where 
$i=\sqrt{-1}$.
Let 
$\bN_0$
be 
$\bN\cup\{0\}$.
For $n\in\bN_0$,
the $n$-th Fock vector, 
or the $n$-th number vector, 
$f_n\in\cH$ 
is defined by 
$f_n
=(1/\sqrt{n!})(\hat a^*)^n\ket0$.
It holds that 
$\hat af_n=\sqrt nf_{n-1}$,
$\hat a^* f_n=\sqrt{n+1}f_{n+1}$
and
$\hat Nf_n=nf_n$.

By Baker-Hausdorff formula,
the displacement operator 
$\hat D_\theta\in\cU(\cH)$
satisfies 
$\hat D_\theta
=e^{-|\theta|^2/2}e^{\theta\hat a^*}e^{-\bar\theta\hat a}$
for $\theta\in\mathbb C$.
Hence
we have
\begin{align}
	\ket\theta
	=
	e^{-|\theta|^2/2}
	\sum_{n=0}^\infty
	\frac{\theta^n}{\sqrt{n!}}
	f_n
	.
	\label{eq_fock_coherent}
\end{align}
Hence,
by the power series calculation,
we have 
\begin{align}
	\langle\theta\mid\eta\rangle
	=e^{-|\theta|^2/2-|\eta|^2/2+\bar\theta\eta}
	.
	\label{eq_inner_product}
\end{align}
Moreover, 
$\{f_n\}_{n=0}^\infty$ 
is a complete orthonormal basis 
of $\cH$ 
because 
$f_n^*g=0$ 
$({^\forall}n\in\bN_0)$ 
is equivalent to 
$\int_\bR e^{i\sqrt2 sx}g(x)dx=0$ 
$({^\forall}s\in\bR)$. 
Hence,
it holds that
$\hat N=\sum_{n=0}^\infty
nf_nf_n^*$.
Calculating the Gaussian mixture of $\ket\theta\bra\theta$ 
using the form of 
$(\ref{eq_fock_coherent})$,
we have 
\begin{align}
	\hat\rho_{0,N}
	=
	\frac{1}{N+1}
	\sum_{n=0}^\infty
	\frac{N^n}{(N+1)^n}
	f_nf_n^*
	.
	\label{eq_gauss_fock}
\end{align}
This equation will be used 
in
Lemma \ref{lem_one_mode_char_func}.

\donotdisplay
{
Hence, 
it holds that 
\[
	\hat\rho_{0,N}
	=
	\begin{cases}
	\frac{1}{2\pi}\int_\bR e^{it\hat N}dt & \mbox{if }N=0,\\
	\frac{1}{N+1}
	\exp\Big[
	\log(\frac{N}{N+1})\hat N
	\Big]
	&\mbox{if }N>0,
	\end{cases}
\]
where 
$r=\log[N/(N+1)]$.
}

For $j,k\in\{1,2,...,n\}$
with $j<k$,
define $K_{j,k}
\in\mathrm{Ant}_\mathbb C^n$
by
\[
	K_{j,k}
	=
	\sqrt{-1}
	\begin{pmatrix}
	O_{j-1}\\
	&1\\
	&&O_{k-j-1}\\
	&&&-1\\
	&&&&O_{n-k}
	\end{pmatrix}
	.
\]
For $k>j$,
let $K_{j,k}=-K_{k,j}$.
Let $K_{j,j}=O_n$.
The operator 
$\hat d_{j,k}$
defined in 
(\ref{eq_hat_d_j_k})
is equal to 
$\hat v_{K_{j,k}}$.
In Sec. \ref{sec_n2},
we will use 
the following lemma. 

\begin{lemma}\label{lem:2_realization}
\footnote{This lemma was suggested by Prof. K. Matsumoto.}
For any $j,k\in\{1,2,..,n\}$,
let 
$\hat U\in\cU(\cH^{\ox mn})$
be 
$\exp((\pi/4)\hat v_{J_{j,k}})$, 
and 
let 
$\hat V\in\cU(\cH^{\ox mn})$
be 
$\exp((\pi/4)\hat v_{K_{j,k}})$.
Then,
it holds that 
$
	\hat U^*\hat V^*
	\hat v_{j,k}
	\hat V\hat U
	=
	\hat d_{j,k}
$.
\end{lemma}
\begin{proof}
Let 
$U=\exp\Big[\frac\pi4\begin{pmatrix}0&-1\\1&0\end{pmatrix}\Big]$, 
and let 
$V=\exp\Big[\frac\pi4\begin{pmatrix}i&0\\0&-i\end{pmatrix}\Big]$. 
As 
$U=\frac{1}{\sqrt2}\begin{pmatrix}1&-1\\1&1\end{pmatrix}$ 
and 
$V=\begin{pmatrix}e^{i\pi/4}&0\\0&e^{-i\pi/4}\end{pmatrix}$, 
it holds that 
$
	U^*V^*
	J_{j,k}
	VU 
	=
	K_{j,k}
$.
By Lemma \ref{lem:v_a_b_c},
we have 
$
	\hat U^*\hat V^*
	\hat v_{j,k}
	\hat V\hat U
	=
	\hat d_{j,k}
$.
\end{proof}

In quantum optical experiments,
the unitary operators 
$\hat U$ and $\hat V$
of the above lemma 
can be realized by 
phase-shifting 
and 
beam-splitting,
respectively.
Moreover,
$\hat T_{j,k}
=-\sqrt{-1}\hat d_{j,k}$
is an observable 
whose POVM can be realized by 
arithmetic subtraction 
of data 
obtained by number measurements.
Hence,
if $n=2$,
then 
$T^\inv_\alpha$
can be realized by 
beam-splitters 
and photon counters.

\subsection{Proof of Theorem \ref{th_n2} (Negative binomial and Poisson distributions)}\label{sec_n2}

Let 
$\hat N
\in\cL(\cH)$
be 
$\hat a^*\hat a$.
Let 
$h_{\theta,N}(r)$
be 
$\Tr[\hat\rho_{\theta,N}e^{i r\hat N}]$,
where 
$i=\sqrt{-1}$.
We write 
$h_{\theta,N}(r)$
using 
$\gamma(r)$ 
and 
$\psi_s(r)$ 
of 
(\ref{eq_def_gamma}).

\begin{lemma}\label{lem_gamma_r}
For any $N\ge0$,
it holds that 
$h_{0,N}(r)
=\gamma(r)$.
\end{lemma}
\begin{proof}
By Eq. $(\ref{eq_gauss_fock})$,
it holds that 
\[
	h_{0,N}(r)
	=
	\frac{1}{N+1}
	\sum_{k=0}^\infty
	\Big(
	\frac
	{e^{i r}N}
	{N+1}
	\Big)^k
	=
	\frac{1}{N+1}
	\frac{1}{1-\frac{e^{i r}N}{N+1}}
	=
	\frac{1}{N+1-Ne^{i r}}
	.
\]
Hence,
we have 
$h_{0,N}(r)
=\gamma(r)$.
\end{proof}

\begin{lemma}\label{lem_h_0}
If $N=0$,
then 
it holds that 
$h_{\theta,0}(r)
=\psi_{|\theta|}(r)$.
\end{lemma}
\begin{proof}
By Eq. (\ref{eq_fock_coherent}),
it holds that 
\[
	h_{\theta,0}(r)
	=
	e^{-|\theta|^2}
	\sum_{k=0}^\infty
	\frac
	{e^{i k r}|\theta|^{2k}}
	{k!}
	=
	\exp[|\theta|^2
	(e^{i r}-1)]
	.
\]
Since 
$\gamma(r)$ 
is one,
$h_{\theta,0}(r)$
is 
$\psi_{|\theta|}(r)$.
\end{proof}

\begin{lemma}\label{lem_one_mode_char_func}
For any $N\ge0$,
it holds that 
$h_{\theta,N}(r)
=\gamma(r)\psi_{|\theta|}(r)$.
\end{lemma}
\begin{proof}
If $N=0$,
then 
$\gamma(r)$
is one,
and 
by 
Lemma \ref{lem_h_0},
we have 
$h_{\theta,0}(r)
=\gamma(r)\psi_{|\theta|}(r)$.
Hereafter,
assume that 
$N>0$.
For 
$x,y\in\bR$,
let $z\in\bC$
be 
$x+iy$,
and let 
$f_{\theta,N}(z)$
be 
$\exp(-|z-\theta|^2/N)$.
It holds that 
\[
	h_{\theta,N}(r)
	=
	\frac
	{1}
	{\pi N}
	\iint_{\bR^2}
	\bra z
	e^{i r\hat N}
	\ket z
	f_{\theta,N}(z)
	dxdy
	.
\]
By Lemma 
\ref{lem_a_b_z},
we have 
$e^{i r\hat N}
\ket z
=\Ket{e^{ir}z}$.
By 
Eq. (\ref{eq_inner_product}),
we have 
$\langle z\ket{e^{i r}z}
=\exp(e^{ir}|z|^2-|z|^2)$.
Hence,
$\bra ze^{i r\hat N}\ket zf_{\theta,N}(z)$
is 
$\exp(-g_{\theta,N}(z))$,
where 
$
	g_{\theta,N}(z)
	=
	|z-\theta|^2/N
	+(1-e^{i r})|z|^2
$.
Let 
$c_{N,r}$ be 
$1/N+1-e^{i r}$.
For $w\in\bC$,
let 
$\vec\mu_w
\in\Mat_\bR^{2,1}$
be 
${^t}(\re(w),\im(w))$.
Then,
it holds that 
\begin{align*}
	g_{\theta,N}(z)
	=&
	c_{N,r}
	\Big(
	{^t}\vec\mu_z
	-\frac{{^t}\vec\mu_\theta}{Nc_{N,r}}
	\Big)
	\Big(
	\vec\mu_z
	-\frac{\vec\mu_\theta}{Nc_{N,r}}
	\Big)
	-\frac{|\theta|^2}{N^2c_{N,r}}
	+\frac{|\theta|^2}{N}
	\\
	=&
	c_{N,r}
	\Big(
	{^t}\vec\mu_z
	-\frac{{^t}\vec\mu_\theta}{Nc_{N,r}}
	\Big)
	\Big(
	\vec\mu_z
	-\frac{\vec\mu_\theta}{Nc_{N,r}}
	\Big)
	-\log(\psi_{|\theta|}(r))
	.
\end{align*}
We can factorize 
$h_{\theta,N}(r)$
to 
$\psi_{|\theta|}(r)$
and 
$h_{0,N}(r)$
as 
\begin{align*}
	h_{\theta,N}(r)
	=&
	\frac{\psi_{|\theta|}(r)}{\pi N}
	\iint_{\bR^2}
	\exp\left[
	-c_{N,r}
	\Big(
	{^t}\vec\mu_z
	-\frac{{^t}\vec\mu_\theta}{Nc_{N,r}}
	\Big)
	\Big(
	\vec\mu_z
	-\frac{\vec\mu_\theta}{Nc_{N,r}}
	\Big)
	\right]
	dxdy
	\\
	=&
	\frac{\psi_{|\theta|}(r)}{\pi N}
	\iint_{\bR^2}
	\exp(
	-c_{N,r}
	{^t}\vec\mu_z
	\vec\mu_z
	)
	dxdy
	=
	\psi_{|\theta|}(r)
	h_{0,N}(r)
	.
\end{align*}
By Lemma \ref{lem_gamma_r},
$h_{0,N}(r)$
is 
$\gamma(r)$.
Hence,
we have 
$h_{\theta,N}(r)=\gamma(r)\psi_{|\theta|}(r)$.
\end{proof}

Next,
we prove 
(i) of 
Theorem \ref{th_n2}
for the case of $m=1$.

\begin{lemma}\label{lem_char_func}
If $m=1$,
and 
if $n=2$,
then 
it holds that 
$\varphi_Y(r)
=
\gamma(r)
\gamma(-r)
\allowbreak
\psi_{|\theta|}(r)
\psi_{|\theta|}(-r)$.
\end{lemma}
\begin{proof}
By Lemma
\ref{lem:dariano},
it holds that
$
	\varphi_Y(r)
	=
	\Tr[\hat\rho_{0,N}\otimes\hat\rho_{\sqrt2\theta,N}
	\exp(r\hat v_{1,2})]
$.
Let 
$\hat U=\exp((\pi/4)\hat v_{J_{1,2}})$, 
and let 
$\hat V=\exp((\pi/4)\hat v_{K_{1,2}})$.
By Lemma 
\ref{lem_g_a_b_z},
it holds that 
$
	\hat U^*
	\hat V^*
	\hat\rho_{0,N}\otimes\hat\rho_{\sqrt2\theta,N}
	\hat V
	\hat U
	=
	\hat\rho_{e^{i\pi/4}\theta,N}^{\ox2}
$.
By 
Lemma \ref{lem:2_realization},
it holds that 
$\hat U^*\hat V^*\hat v_{1,2}\hat V\hat U
=\hat d_{j,k}$.
Hence,
we have 
\begin{align*}
	\varphi_Y(r)
	=&
	\Tr[\hat\rho_{e^{i\pi/4}\theta,N}^{\ox2}\exp(r\hat d_{1,2})]
	\\
	=&
	\Tr[\hat\rho_{e^{i\pi/4}\theta,N}\exp(i r\hat N)]
	\Tr[\hat\rho_{e^{i\pi/4}\theta,N}\exp(-i r\hat N)]
	.
\end{align*}
By Lemma 
\ref{lem_one_mode_char_func},
we have 
$
	\varphi_Y(r)
	=
	\gamma(r)\gamma(-r)
	\psi_{|\theta|}(r)
	\psi_{|\theta|}(-r)
$.
\end{proof}

For any random variable $Z$, 
and for any constant $c\in\bR$, 
the characteristic functions 
of $Z$ and $cZ$
satisfy 
\begin{align}
	\varphi_{cZ}(t)=\varphi_Z(ct)
	.
	\label{eq_c_char}
\end{align}
For any mutually independent random variables 
$Z$ and $W$,
the characteristic functions of 
$Z$, $W$ and $S=Z+W$ satisfy 
\begin{align}
	\varphi_S(r)
	=
	\varphi_Z(r)
	\varphi_W(r)
	.
	\label{eq_char_ind}
\end{align}

\begin{proof}[{\bf Proof of Theorem \ref{th_n2}}]
(i) 
Let $\theta_k$ 
be the $k$-th entry of 
$\vec\theta\in\Mat_\bC^{m,1}$.
By Lemma 
\ref{lem_char_func}, 
and by 
Eq. 
(\ref{eq_char_ind}),
we have 
\begin{align*}
	\varphi_Y(r)
	=&
	[\gamma(r)\gamma(-r)]^m
	\Big[
	\prod_{k=1}^m
	\psi_{|\theta_k|}(r)
	\psi_{|\theta_k|}(-r)
	\Big]
	\\
	=&
	[\gamma(r)\gamma(-r)]^m
	\psi_{\|\vec\theta\|}(r)
	\psi_{\|\vec\theta\|}(-r)
	.
\end{align*}
(ii) 
Let 
$S=F-G+\sum_{k=1}^\infty(kP_k-kQ_k)$. 
By Eqs. (\ref{eq_c_char})
and (\ref{eq_char_ind}),
it holds that 
\begin{align*}
	\varphi_S(r)
	=&
	\varphi_F(r)\varphi_G(-r)
	\prod_{k=1}^\infty
	[
	\varphi_{P_k}(kr)
	\varphi_{Q_k}(-kr)
	]
	.
\end{align*}
Since $F$ and $G$ obey 
$NB_m(N(N+1)^{-1})$,
it holds that 
$\varphi_F(r)
=\varphi_G(r)
=\gamma(r)$. 
Let 
$\lambda_k=\|\vec\theta\|^2N^{k-1}(N+1)^{-k-1}$.
Since 
$P_k$ and $Q_k$ obey 
$\Poi(\lambda_k)$,
it holds that 
$\varphi_{P_k}(r)
=\varphi_{Q_k}(r)
=\exp[\lambda_k(e^{i r}-1)]$. 
As 
\begin{align*}
	&
	\sum_{k=1}^\infty\log[\varphi_{P_k}(kr)]
	=
	\frac{\|\vec\theta\|^2}{N(N+1)}
	\sum_{k=1}^\infty
	\Big(
	\frac{N}{N+1}
	\Big)^k
	(e^{i k r}-1)
	\\
	&=
	\frac{\|\vec\theta\|^2}{N(N+1)}
	\Big(
	\frac{1}{1-\frac{Ne^{i r}}{N+1}}
	-
	\frac{1}{1-\frac{N}{N+1}}
	\Big)
	=
	\frac{\|\vec\theta\|^2}{N}
	\Big(
	\frac{1}{N+1-Ne^{i r}}
	-
	1
	\Big)
\end{align*}
holds,
we have 
$\varphi_S(r)
=\gamma(r)
\gamma(-r)
\psi_{\|\vec\theta\|}(r)
\psi_{\|\vec\theta\|}(-r)$.
\end{proof}

\subsection{Proof of Theorem \ref{th_type2} (Type II error probability of $T^\inv_\alpha$ for $N=0$)}\label{sec_proof_type2}

For 
$\mu,\nu\in\bN$,
let 
$T^{\mu,\nu}_X$
be the tangent space of 
$\Mat_\bR^{\mu,\nu}$
at 
$X\in\Mat_\bR^{\mu,\nu}$.
Let 
$x_{i,j}$
be the coordinate  variable 
for the $(i,j)$-th entry 
of 
$X\in\Mat_\bR^{\mu,\nu}$.
Then,
a basis of 
$T^{\mu,\nu}_X$
is 
$\{\partial/\partial x_{i,j}|_X
\mid1\le i\le\mu,1\le j\le\nu\}$.
For 
$Y\in\Mat_\bR^\mu$
and 
for 
$Z\in\Mat_\bR^\nu$,
let 
$f_{Y,Z}:\Mat_\bR^{\mu,\nu}\to\Mat_\bR^{\mu,\nu}$
be a map given by 
$X\mapsto YX({^t}Z)$.
The pushforward 
$df_{Y,Z}|_X:T^{\mu,\nu}_X\to T^{\mu,\nu}_{f_{Y,Z}(X)}$
at $X\in\Mat_\bR^{\mu,\nu}$
is given by 
$df_{Y,Z}|_X(\partial/\partial x_{i,j}|_X)
=\sum_{m=1}^\mu\sum_{n=1}^\nu
Y_{m,i}Z_{n,j}
\partial/\partial x_{m,n}|_{f_{Y,Z}(X)}$,
where 
$M_{k,l}$
is the $(k,l)$-th entry of 
a matrix $M$.
For $X\in\Mat_\bR^{\mu,\nu}$
and for $K\in\bN$,
let 
$\cP
=\{P_1(X),P_2(X),...,P_K(X)\}
\subset\Mat_\bR^\mu$
and 
$\cQ
=\{Q_1(X),Q_2(X),...,Q_K(X)\}
\subset\Mat_\bR^\nu$
be sets of 
$K$ matrices 
which satisfy 
$\sum_{k=1}^K
\tr_\nu[{^t}MP_k(X)MQ_k(X)]>0$
for any non-zero matrix 
$M\in\Mat_\bR^{\mu,\nu}$.
The inner product 
$(u,v)^{\cP,\cQ}_X$
of 
$u,v\in T^{\mu,\nu}_X$
is defined by 
$(\partial/\partial x_{p,q}|_X,\partial/\partial x_{r,s}|_X)^{\cP,\cQ}_X
=\sum_{k=1}^KP_{k,p,r}(X)Q_{k,s,q}(X)$,
where 
$P_{k,i,j}(X)$
and 
$Q_{k,i,j}(X)$
are the $(i,j)$-th entries of 
$P_k(X)$ and $Q_k(X)$,
respectively.
Let 
$\|u\|^{\cP,\cQ}_X$
be 
$\sqrt{(u,u)^{\cP,\cQ}_X}$.
If 
$\cP=\{I_\mu\}$
and 
$\cQ=\{I_\nu\}$
for any $X\in\Mat_\bR^{\mu,\nu}$,
then 
the inner product 
is said to be 
Euclidean.
The pullback of 
$(\cdot,\cdot)^{\cP,\cQ}_{f_{Y,Z}(X)}$
by 
$f_{Y,Z}$
is 
$(\cdot,\cdot)^{\cP',\cQ'}_X$,
where 
$\cP'\subset\Mat_\bR^\mu$
and 
$\cQ'\subset\Mat_\bR^\nu$
consist of 
${^t}YP_k(f_{Y,Z}(X))Y$
and 
${^t}ZQ_k(f_{Y,Z}(X))Z$,
respectively.

Let 
$\SO_\bR^n$
be 
$\{e^A\mid A\in\Ant_\bR^n\}$,
the set of 
special orthogonal matrices.
For $R\in\SO_\bR^n$,
let 
$T_R\subset T^{n,n}_R$
be the tangent space 
of 
$\SO_\bR^n\subset\Mat_\bR^n$
at $R$.
We use 
$(\cdot,\cdot)^{\cP,\cQ}_R$
by restricting the space 
$T^{n,n}_R$
to $T_R$.
The dimension of 
$T_R$
is 
$d_n=n(n-1)/2$.
If $C\subset T_R$
is a cuboid 
framed by edge vectors 
$v_1,v_2,...,v_{d_n}\in T_R$
which are mutually orthogonal 
with respect to 
$(\cdot,\cdot)^{\cP,\cQ}_R$,
then the volume 
$\vol^{\cP,\cQ}_R[C]$
is 
$\prod_{k=1}^{d_n}
\|v_k\|^{\cP,\cQ}_R$.
Let 
$\cF_{\SO_\bR^n}$
be the set of Borel subsets of 
$\SO_\bR^n$.
Let 
$\mu^{\cP,\cQ}:\cF_{\SO_\bR^n}\to[0,\infty)$
be a measure on $\SO_\bR^n$
given by 
$\vol^{\cP,\cQ}_R[\cdot]$.
For $L\in\SO_\bR^n$,
and for 
$M\subset\SO_\bR^n$,
let $LM\subset\SO_\bR^n$
be 
$\{Lx\mid x\in M\}$.
The following lemma implies that, 
if 
the inner product is Euclidean,
then 
$\mu^{\cP,\cQ}$ is a left-invariant Haar measure.

\begin{lemma}\label{lem_left_inv}
If 
$\cP=\cQ=\{I_n\}$,
then,
for any $L\in\SO_\bR^n$,
and for any 
$M\in\cF_{\SO_\bR^n}$
it holds that 
$\mu^{\cP,\cQ}(LM)
=\mu^{\cP,\cQ}(M)$.
\end{lemma}
\begin{proof}
Let 
$f_L:\SO_\bR^n\to\SO_\bR^n$
be 
$U\mapsto LU$,
and let 
$df_L|_U:T_U\to T_{LU}$
be the pushforward.
For any $U\in\SO_\bR^n$
and 
for any $u,v\in T_U$,
we have 
$
	(df_L(u),df_L(v))^{\{I_n\},\{I_n\}}_{LU}
	=(u,v)^{\{I_n\},\{I_n\}}_U
$.
Let 
$C\subset M$
be any infinitesimal cuboid,
and let 
$R\in C$ be one of the vertices.
Then,
$C$
is identified by 
mutually orthogonal edge vectors 
$v_1,v_2,...,v_{d_n}\in T_R$.
It holds that 
$\vol^{\cP,\cQ}_{LR}[LC]
=\vol^{\cP,\cQ}_R[C]$.
Hence,
we have 
$\mu^{\cP,\cQ}(LM)
=\mu^{\cP,\cQ}(M)$.
\end{proof}

Let 
$\mu:\cF_{\SO_\bR^n}\to[0,\infty)$
be 
$\mu^{\cP,\cQ}$
with 
$\cP=\cQ=\{I_n\}$.
For $U\in\SO_\bR^n$,
let 
$\dhat V_U\in\cU(\cH^{\ox mn})$
be 
$\exp(\hat v_{\log U})$.
By Lemma \ref{lem:v_a_b_c},
for any 
$L,R\in\SO_\bR^n$,
it holds that 
$\dhat V_L\dhat V_R
=\dhat V_{L R}$.
Moreover,
by Lemma \ref{lem_v_j_k_f_x},
it holds that 
$\dhat V_L\dhat V_Rf(X)
=\dhat V_Lf(X R)
=f(X L R)$.
Hence, 
$\SO_\bR^n$
acts on the left of $\cH^{\ox mn}$
by 
$\dhat V:\SO_\bR^n\to\cU(\cH^{\ox mn})$.
Let 
$
	\hat W
	\in
	\cL(\cH^{\ox mn})
$
be 
$
	\mu(\SO_\bR^n)^{-1}
	\int_{\SO_\bR^n}
	\dhat V_U
	\mu(dU)
$.

\begin{lemma}\label{lem_proj}
For any mode size $m\ge1$,
it holds that 
$
	\hat K_0
	=
	\hat W
$.
\end{lemma}
\begin{proof}
For any $f\in\cH^{\ox mn}$,
and for any $L\in\SO_\bR^n$,
it holds that 
$
	\dhat V_L\hat Wf
	=
	\mu(\SO_\bR^n)^{-1}
	\int_{\SO_\bR^n}
	f(X L U)
	\mu(dU)
$.
Let 
$U'$
be 
$L U$.
By 
Lemma \ref{lem_left_inv},
it holds that 
$
	\mu(dU)
	=
	[
	\mu(dU)
	/
	\mu(dU')
	]
	\mu(dU')
	=
	\mu(dU')
$.
Hence, 
it holds that 
$\dhat V_L\hat Wf=\hat Wf$.
By Theorem \ref{th_radial},
we have 
$\hat Wf\in\cK_0$.
If 
$g\in\cK_0$,
then 
it holds that 
$\dhat V_Ug=g$
for any $U\in\SO_\bR^n$.
Hence,
we have 
$
	\hat Wg
	=
	\mu(\SO_\bR^n)^{-1}
	\int_{\SO_\bR^n}
	g
	\mu(dU)
	=
	g
$.
As a result, 
$\hat W$
is the projection 
on $\cK_0$.
\end{proof}

Let 
$\cS^{n-1}\subset\Mat_\bR^{1,n}$
be 
the set of unit row vectors.
For $M\subset\cS^{n-1}$,
and for $U\in\SO_\bR^n$,
let 
$MU\subset\cS^{n-1}$ 
be 
$\{\vec vU\mid\vec v\in M\}$.
Let $\cF_{\cS^{n-1}}$
be the set of Borel subsets of 
$\cS^{n-1}$.
If a measure 
$\lambda:\cF_{\cS^{n-1}}\to[0,\infty)$
satisfies 
$\lambda(MU)=\lambda(M)$
(${^\forall}M\in\cF_{\cS^{n-1}}$,
${^\forall}U\in\SO_\bR^n$),
then 
$\lambda$ 
is said to be 
rotationally uniform.
For 
$\vec v\in\cS^{n-1}$,
let 
$T^\cS_{\vec v}\subset T^{1,n}_{\vec v}$
be the tangent space of 
$\cS^{n-1}\subset\Mat_\bR^{1,n}$
at $\vec v$.
We use 
the inner product 
$(\cdot,\cdot)^{\cP,\cQ}_{\vec v}$
by restricting the space 
to $T^\cS_{\vec v}$.
Since 
elements of 
$\cP$ are scalars,
we assume, 
without loss of generality,
that 
$\cP=\{I_1\}$
and that 
$\cQ=\{Q(\vec v)\}$,
where 
$Q(\vec v)\in\Sym_\bR^n$
is a positive matrix 
which may depend on 
$\vec v$.
Let 
$(\cdot,\cdot)^\cQ_{\vec v}
=(\cdot,\cdot)^{\{I_1\},\{Q(\vec v)\}}_{\vec v}$,
and let 
$\|\cdot\|^\cQ_{\vec v}
=\|\cdot\|^{\{I_1\},\{Q(\vec v)\}}_{\vec v}$.
The dimension of 
$T^\cS_{\vec v}$
is 
$n-1$.
If 
$C\in T^\cS_{\vec v}$
is a cuboid with edge vectors 
$v_1,v_2,...,v_{n-1}
\in T^\cS_{\vec v}$
which are mutually orthogonal 
with respect to 
$(\cdot,\cdot)^\cQ_{\vec v}$.
then 
the volume 
$\vol^\cQ_{\vec v}[C]$
is 
$\prod_{k=1}^{n-1}\|v_k\|^\cQ_{\vec v}$.
Let 
$\lambda^\cQ:\cF_{\cS^{n-1}}\to[0,\infty)$
be the measure given by 
$\vol^\cQ_{\vec v}[\cdot]$.

\begin{lemma}\label{lem_rot_uni_euc}
If 
$\lambda^\cQ:\cF_{\cS^{n-1}}\to[0,\infty)$
is rotationally uniform,
then there exists 
$c>0$ such that, 
for any $\vec v\in\cS^{n-1}$,
$Q(\vec v)$
is 
$cI_n$.
\end{lemma}
\begin{proof}
For 
$U\in\SO_\bR^n$,
define 
$f_U:\cS^{n-1}\to\cS^{n-1}$
by 
$\cS^{n-1}\ni\vec v\mapsto\vec vU^{-1}=\vec v({^t}U)$.
For any 
$U\in\SO_\bR^n$,
for any $\vec v\in\cS^{n-1}$,
and for any 
$u,v\in T^\cS_{\vec v}$,
it holds that 
$(df_U(u),df_U(v))^{\{Q(\vec v)\}}_{\vec vU^{-1}}
=(u,v)^{\{{^t}UQ(\vec v)U\}}_{\vec v}$.
The group action 
$(U,\vec v)\mapsto f_U(\vec v)$
is transitive,
that is,
for any $\vec v\in\cS^{n-1}$,
the orbit 
$\{f_U(\vec v)\mid U\in\SO_\bR^n\}$
is $\cS^{n-1}$.
For 
$\vec v\in\cS^{n-1}$,
let 
$\Sigma_{\vec v}\subset\SO_\bR^n$
be 
$\{U\in\SO_\bR^n\mid f_U(\vec v)=\vec v\}$,
the set of stabilizers of 
$\vec v$.
For any $\vec v\in\cS^{n-1}$,
$\Sigma_{\vec v}$
is isomorphic to 
$\SO_\bR^{n-1}$.
For 
$\vec v\in\cS^{n-1}$,
define a group representation 
$\rho_{\vec v}:\Sigma_{\vec v}\to\GL(T^\cS_{\vec v})$
by 
$U\mapsto df_U|_{\vec v}$.
Then,
$\rho_{\vec v}$
is irreducible.
Since 
$\lambda^\cQ$ 
is rotationally uniform,
it is necessary that, 
for any 
$\vec v\in\cS^{n-1}$,
for any 
$R\in\Sigma_{\vec v}$,
and for any 
$u,v\in T^\cS_{\vec v}$,
$(\rho_{\vec v}(R)u,\rho_{\vec v}(R)v)^\cQ_{\vec v}$
is 
$(u,v)^\cQ_{\vec v}$,
and so that, 
for any 
$\vec v\in\cS^{n-1}$,
and for any 
$R\in\Sigma_{\vec v}$,
${^t}RQ(\vec v)R$
is 
$Q(\vec v)$.
By Schur's lemma,
there exists 
a function 
$\varphi:\cS^{n-1}\to\bR$
such that 
$Q(\vec v)=\varphi(\vec v)I_n$.
(See \cite{fulton}.)
By the positivity of 
the inner product,
$\varphi(\vec v)$
is positive  
for any 
$\vec v\in\cS^{n-1}$.
By the transitivity 
of $f_U$,
$\varphi(\vec v)$
is constant.
\end{proof}

For 
$\vec v\in\cS^{n-1}$,
define a measure 
$\lambda_{\vec v}:\cF_{\cS^{n-1}}\to[0,\infty)$
by 
$
	\lambda_{\vec v}(M)
	=
	\mu\big(
	\{
	U\in\SO_\bR^n\mid
	\vec vU^{-1}
	\allowbreak
	\in M
	\}
	\big)
$.

\begin{lemma}\label{lem_uniform}
For any $\vec v\in\cS^{n-1}$,
$\lambda_{\vec v}$
is rotationally uniform.
\end{lemma}
\begin{proof}
Choose 
$M\in\cF_{\cS^{n-1}}$
and 
$U\in\SO_\bR^n$,
arbitrarily.
It holds that 
$
	\{L\in\SO_\bR^n\mid\vec vL^{-1}\in MU\}
	=
	U^{-1}\{R\in\SO_\bR^n\mid\vec vR^{-1}\in M\}
$.
By Lemma \ref{lem_left_inv},
we have 
$\lambda_{\vec v}(MU)
=\mu(U^{-1}\{R\mid\vec vR^{-1}\in M\})
=\mu(\{R\mid\vec vR^{-1}\in M\})
=\lambda_{\vec v}(M)$.
\end{proof}

Let 
$\lambda_1:\cF_{\cS^{n-1}}\to[0,\infty)$
be the rotationally uniform measure 
with $\lambda_1(\cS^{n-1})=1$.
Consider the case of $m=1$,
and,
for 
$\vec r
=(r_1,r_2,...,r_n)
\in\Mat_\bR^{1,n}$,
let 
$\Ket{\vec r}\in\cH^{\ox n}$
be 
$\ket{r_1}\ox\ket{r_2}\ox\cdots\ox\ket{r_n}$.

\begin{lemma}\label{lem_lambda_1}
If the mode size $m$ is one,
then,
for any $r\ge0$,
and 
for any $\vec u\in\cS^{n-1}$,
it holds that 
$\hat K_0\Ket{r\vec u}
=\int_{\cS^{n-1}}
\Ket{r\vec v}
\lambda_1(d\vec v)$.
\end{lemma}
\begin{proof}
By Lemma \ref{lem_proj},
it holds that 
$\hat K_0\Ket{r\vec u}
=
\mu(\SO_\bR^n)^{-1}
\int_{\SO_\bR^n}
\dhat V_U\Ket{r\vec u}
\mu(dU)$.
By Lemma \ref{lem_a_b_z},
it holds that 
$\dhat V_U\Ket{r\vec u}
=\Ket{r\vec uU^{-1}}$,
and so that 
$\hat K_0\Ket{r\vec u}
=
\mu(\SO_\bR^n)^{-1}
\int_{\cS^{n-1}}
\Ket{r\vec v}
\lambda_{\vec u}(d\vec v)$.
By Lemma \ref{lem_uniform},
we have 
$\hat K_0\Ket{r\vec u}
=
\int_{\cS^{n-1}}
\Ket{r\vec v}
\allowbreak
\lambda_1(d\vec v)$.
\end{proof}

Let 
$\bT\subset\bR^{n-2}$
be 
$[0,\pi)^{n-2}$,
and let 
$\cT\subset\bR^{n-1}$
be 
$\bT\times[0,2\pi)$.
For 
$k\in\{1,2,...,n-1\}$,
let 
$t_k$
be the coordinate variable 
of the $k$-th entry of 
$\vec t\in\cT$.
We define 
unit row vectors 
$\vec\nu_0,\vec\nu_1,...,\vec\nu_{n-1}
\in\cS^{n-1}$
which are 
parameterized by 
$\vec t\in\cT$,
as follows.
Let 
$\vec\nu_0\in\cS^{n-1}$
be 
$(1,O_{1,n-1})$.
For 
$k\in\{1,2,...,n-1\}$,
let 
$\vec\nu_k\in\cS^{n-1}$
be 
\[
	\Big(
	\cos t_1,\
	\cos t_2\sin t_1,\
	\cdots,\
	\cos t_k\prod_{j=1}^{k-1}\sin t_j,\
	\prod_{j=1}^k\sin t_j,
	O_{1,n-k-1}
	\Big)
	.
\]
Define 
$\varphi:\cT\to\cS^{n-1}$
by 
$\varphi(\vec t)=\vec\nu_{n-1}$.
For 
$k\in\{1,2,...,n\}$,
let 
$x_k$
be the coordinate variable 
of the $k$-th entry of 
$\vec v\in\cS^{n-1}$.
The pushforward 
$d\varphi(\partial/\partial t_i|_{\vec t})$
is 
$\sum_{j=1}^n
(\partial x_j/\partial t_i)
\partial/\partial x_j
|_{\varphi(\vec t)}$.
Let 
$\vec J_i\in\Mat_\bR^{1,n}$
be a row vector whose 
$j$-th entry is 
$\partial x_j/\partial t_i$.
Then,
as the pullback of 
the Euclidean inner product for 
$T^\cS_{\varphi(\vec t)}$,
the inner product of 
$\partial/\partial t_i|_{\vec t}$
and 
$\partial/\partial t_j|_{\vec t}$
is given by 
\[
	\Big(
	\frac{\partial}{\partial t_i}\Big|_{\vec t},
	\frac{\partial}{\partial t_j}\Big|_{\vec t}
	\Big)_{\vec t}
	=
	\Big(
	d\varphi\big(
	\frac{\partial}{\partial t_i}
	\Big|_{\vec t}
	\big)
	,
	d\varphi\big(
	\frac{\partial}{\partial t_j}
	\Big|_{\vec t}
	\big)
	\Big)^{\{I_n\}}_{\vec\nu_{n-1}}
	=
	\vec J_i
	({^t}\vec J_j)
	.
\]
For $k\in\{1,2,...,n-1\}$,
let
\[
	R_k
	=
	\begin{pmatrix}
	I_{k-1}\\
	&\cos t_k&\sin t_k\\
	&-\sin t_k&\cos t_k\\
	&&&I_{n-k-1}
	\end{pmatrix}
	\in
	\Mat_\bR^n
	,
\]
and 
let
$
	D_k
	=\partial R_k/\partial t_k
	\in
	\Mat_\bR^n
$.
Let
$R_{j,k}=R_jR_{j+1}\cdots R_k$. 
If $j>k$, 
let 
$R_{j,k}$
be 
$I_n$. 
It holds that
$\vec\nu_k=\vec\nu_0R_{1,k}$
and that 
$
	\vec J_i
	=
	\vec\nu_{i-1}
	D_i
	R_{i+1,n-1}
$.

\begin{lemma}\label{lem_e_k}
It holds that
$
	\|d\varphi(\partial/\partial t_k|_{\vec t})\|^{\{I_n\}}_{\vec\nu_{n-1}}
	=
	c_k
$,
where
\[
	c_k
	=
	\begin{cases}
	1&\mbox{if }k=1,\\
	\prod_{j=1}^{k-1}\sin t_j
	&\mbox{if }2\le k\le n-1
	.
	\end{cases}
\]
\end{lemma}
\begin{proof}
Define 
$E_k\in\Mat_\bR^n$ by
\[
	E_k
	=
	\begin{pmatrix}
	O_{k-1}\\
	&I_2\\
	&&O_{n-k-1}
	\end{pmatrix}
	.
\]
It holds that
$D_k({^t}D_k)=E_k$.
Since the $k$-th entry of 
$\vec\nu_{k-1}$ is $c_k$,
it holds that
$
	\vec\nu_{k-1}E_k({^t}\vec\nu_{k-1})
	=
	c_k^2
$.
Hence, 
we have
$
	\vec J_k({^t}\vec J_k)
	=
	c_k^2
$.
Since, 
for $j<k$,
$t_j$
belongs to 
$[0,\pi)$,
it holds that 
$\sin t_j\ge0$,
and so that 
$\sqrt{c_k^2}=c_k$.
\end{proof}

\begin{lemma}\label{lem_tanaka}\label{lem_f_k}
\footnote{This lemma was suggested by Prof. F. Tanaka.}
If $j\ne k$,
then 
$d\varphi(\partial/\partial t_j|_{\vec t})$
is orthogonal to 
$d\varphi(\partial/\partial t_k|_{\vec t})$
with respect to 
the Euclidean inner product.
\end{lemma}
\begin{proof}
Define 
$A_k\in\Ant_\bR^n$ 
and 
$F_k\in\Mat_\bR^n$ 
by 
\[
	A_k
	=
	\begin{pmatrix}
	O_{k-1}\\
	&0&-1\\
	&1&0\\
	&&&O_{n-k-1}
	\end{pmatrix}
	\mbox{ and }
	F_k
	=
	\begin{pmatrix}
	I_k\\
	&O_{n-k}
	\end{pmatrix}
	,
\]
respectively.
Then,
it hoolds that 
$R_k({^t}D_k)=A_k$
and that 
$F_kA_kF_k=O_n$.
Assume that $j<k$.
Let 
$\vec\nu_{j,k}
=\vec\nu_{j-1}D_jR_{j+1,k}$.
Then, 
it holds that
$\vec J_j({^t}\vec J_k)
=
	\vec\nu_{j,k-1}
	A_k
	({^t}\vec\nu_{k-1})
$,
and that 
\[
	\vec\nu_{j,k-1}
	A_k
	({^t}\vec\nu_{k-1})
	=
	(\vec\nu_{j,k-1}
	F_k)
	A_k
	(F_k
	{^t}\vec\nu_{k-1})
	=
	\vec\nu_{j,k-1}
	O_n
	{^t}\vec\nu_{k-1}
	=
	0
	.
\]
Hence,
we have
$\vec J_j({^t}\vec J_k)
=0$.
\end{proof}

Let 
$P\subset T^\cS_{\vec\nu_{n-1}}$
be a parallelepiped 
framed by 
$\{
	d\varphi(\partial/\partial t_k|_{\vec t})
	\mid k=1,2,...,n-1\}
$.
Let 
\[
	s_k
	=
	\begin{cases}
	\int_0^\pi\sin^{n-k-1}tdt&\mbox{if }1\le k\le n-2,\\
	2\pi&\mbox{if }k=n-1.
	\end{cases}
\]

\begin{lemma}\label{lem_vol_parallelepiped}
If the inner product of 
$T^\cS_{\vec\nu_{n-1}}$
is Euclidean,
then 
\\
(i)\
the volume of 
$P$ 
is 
$\prod_{k=1}^{n-2}
\sin^{n-k-1}t_k$,
\\
and 
\\
(ii)\
the area 
$|\cS^{n-1}|$
of 
$\cS^{n-1}$
is 
$\prod_{k=1}^{n-1}s_k$,
\end{lemma}
\begin{proof}
(i)\
By Lemma 
\ref{lem_tanaka}\label{lem_f_k},
$P$
is a cuboid 
with respect to 
the Euclidean inner product.
By Lemmas 
\ref{lem_e_k}, 
we have 
$
	\vol^{\{I_n\}}_{\vec\nu_{n-1}}[P]
	=
	\prod_{k=1}^{n-2}
	\sin^{n-k-1}t_k
$.
\\
(ii)\
By taking integral of (i),
we obtain (ii).
\end{proof}

For $r\ge0$, 
define 
$g_r\in\cH^{\ox n}$ 
by 
\[
	g_r
	=
	\frac{1}{|\cS^{n-1}|}
	\int_0^{2\pi}
	\Big[
	\left.\overbrace{\int\cdots\int}^{n-2}\right._{\llap{\ \scriptsize$\bT$}}
	\Ket{r\vec\nu_{n-1}}
	\prod_{k=1}^{n-2}
	(\sin t_k)^{n-k-1}
	dt_k
	\Big]
	dt_{n-1}
	,
\]

\begin{lemma}\label{lem_k0}
If the mode size $m$ is one,
then 
for any $\vec r\in\Mat_\bR^{1,n}$, 
it holds that 
$\hat K_0\Ket{\vec r}
=g_{\|\vec r\|}$.
\end{lemma}
\begin{proof}
By Lemma 
\ref{lem_lambda_1},
we have 
$\hat K_0\Ket{\vec r}
=\int_{\cS^{n-1}}
\Ket{\|\vec r\|\vec v}
\lambda_1(d\vec v)$.
By Lemma 
\ref{lem_rot_uni_euc},
the integration by 
$\lambda_1$
is calculated by 
the Euclidean inner product.
By Lemma 
\ref{lem_vol_parallelepiped},
we have 
$\hat K_0\Ket{\vec r}
=g_{\|\vec r\|}$.
\donotdisplay
{
First,
define the set of orthogonal matrices
with determinant one by  
$\SO_\bR^n=\{e^A\mid A\in\Ant_\bR^n\}$.
Let 
$\mu$ be a left-invariant Haar measure 
on $\SO_\bR^n$.
By Theorem \ref{th_radial} 
with $m=1$, 
the projection 
$\hat K_0:\cH^{\ox n}\to\cK_0$
is represented by 
\[
	(\hat K_0 f)(\vec x)
	=
	\frac{1}{\mu(\SO_\bR^n)}
	\int_{\SO_\bR^n}
	f(\vec x U)
	\mu(dU)
	.
\]
By Eq.
(\ref{eq_def_coherent}),
we have 
\[
	\Ket{\vec r}(\vec x)
	=
	\frac
	{1}
	{\pi^{1/4}}
	e^{-\vec r\vec r^*/2}
	e^{-\vec x({^t}\vec x)/2}
	e^{-\vec r({^t}\vec r)/2}
	e^{\sqrt2\vec x({^t}\vec r)}
	.
\]
Hence,
it holds that 
\begin{align*}
	(\hat K_0\Ket{\vec r})(\vec x)
	=&
	\frac{1}{\mu(\SO_\bR^n)}
	\int_{\SO_\bR^n}
	\Ket{\vec r({^t}U)}(\vec x)
	\mu(dU)
	.
\end{align*}
Expanding the integral 
explicitly,
we obtain 
$\hat K_0\Ket{\vec r}=g_{\|\vec r\|}$.
}
\end{proof}

\donotdisplay
{
\begin{lemma}\label{lem:int_sin}
Assume that $m=1$.
For
$m\in\mathbb N$,
it holds that
\begin{align*}
	\int_0^{\pi}
	\sin^m t dt
	=&
	B\Big(
	\frac{m+1}{2},\frac{1}{2}
	\Big)
	,
\end{align*}
where $B(x,y)$ is the beta function.
\end{lemma}
\begin{proof}
As
$\sin(t-\pi/2)$
is an even function,
it holds that
$\int_0^\pi\sin^mt dt
=$
$2\int_0^{\pi/2}$
$\sin^mt dt$.
Replace $\sin^2t$
by $s$.
The integration interval is
$0\le s\le1$,
and the Jacobian is
\[
	\frac{dt}{ds}
	=
	\frac{1}{ds/dt}
	=
	\frac{1}{2\cos t\sin t}
	=
	\frac{1}{2\sqrt{1-s}\sqrt{s}}
	.
\]
Hence,
\begin{align*}
	&
	\int_0^\pi
	\sin^m t dt
	=
	2\int_0^1
	\frac
	{s^{m/2}}
	{2\sqrt{1-s}\sqrt s}
	ds
	\\
	&
	=
	\int_0^1
	s^{(m+1)/2-1}(1-s)^{1/2-1}
	ds
	=
	B\Big(\frac{m+1}{2},\frac{1}{2}\Big)
\end{align*}
is obtained.
\end{proof}
}

\donotdisplay
{
\begin{lemma}\label{lem_vec_r_vec_s}
Assume that $m=1$.
For any $\vec r,\vec s\in\Mat_\bR^{1,n}$,
it holds that 
\[
	\Bra{\vec r}\hat K_0\Ket{\vec s}
	=
	\Bra{\|\vec r\|\vec\nu_0}\hat K_0\Ket{\vec s}
	.
\]
\end{lemma}
\begin{proof}
Let $A\in\Ant_\bR^n$
be a solution to the equation 
$\vec\nu_0e^{-A}=\vec r/\|\vec r\|$.
By Lemma \ref{lem_a_b_z},
it holds that 
$
	\Bra{\vec r}\hat K_0\Ket{\vec s}
	=
	\Bra{\|\vec r\|\vec\nu_0}e^{-\hat v_A}\hat K_0\Ket{\vec s}
$.
By Theorem \ref{th_radial},
it holds that 
$
	\Bra{\|\vec r\|\vec\nu_0}e^{-\hat v_A}\hat K_0\Ket{\vec s}
	=
	\Bra{\|\vec r\|\vec\nu_0}\hat K_0\Ket{\vec s}
$.
\end{proof}
}

\begin{lemma}\label{lem_type_ii_error_prob}
Assume that $m=1$.
For $\vec r\in\Mat_\bR^{1,n}$,
it holds that
\begin{align}
	\Bra{\vec r}\hat K_0\Ket{\vec r}
	=
	\frac{e^{-\|\vec r\|^2}}{B\left(\frac{n-1}{2},\frac12\right)}
	\int_0^\pi
	e^{\|\vec r\|^2\cos\varphi}
	(\sin\varphi)^{n-2}
	d\varphi
	.
	\label{eq_type_ii_error_prob}
\end{align}
\end{lemma}
\begin{proof}
\donotdisplay
{
By Lemma \ref{lem_vec_r_vec_s},
it holds that 
$
	\Bra{\vec r}\hat K_0\Ket{\vec r}
	=
	\Bra{\|\vec r\|\vec\nu_0}\hat K_0\Ket{\vec r}
$.
}
By Eq. $(\ref{eq_inner_product})$,
for any $r\ge0$,
it holds that 
\[
	\langle r\vec\nu_0\mid r\vec\nu_{n-1}\rangle
	=
	e^{-r^2+r^2\cos t_1}
	.
\]
By Lemma \ref{lem_k0}, 
it holds that 
\begin{align*}
	&
	\Bra{\vec r}
	\hat K_0
	\Ket{\vec r}
	=
	\Bra{\|\vec r\|\vec\nu_0}
	\hat K_0
	\Ket{\vec r}
	\\
	&=
	\frac{1}{|\cS^{n-1}|}
	\int_0^{2\pi}
	\Big[
	\left.\overbrace{\int\cdots\int}^{n-2}\right._{\llap{ \scriptsize$\bT$}}
	e^{-\|\vec r\|^2+\|\vec r\|^2\cos t_1}
	\prod_{k=1}^{n-2}
	\sin^{n-k-1}t_k
	dt_k
	\Big]
	dt_{n-1}
	\\
	&=
	\frac{1}{s_1}
	\int_0^\pi
	e^{-\|\vec r\|^2+\|\vec r\|^2\cos t_1}
	\sin^{n-2}t_1
	dt_1
	.
\end{align*}
Calculating 
the beta integration 
$
	s_1
	=\int_0^\pi\sin^{n-2}tdt
$,
we obtain 
(\ref{eq_type_ii_error_prob}).
\end{proof}

\begin{proof}[{\bf Proof of Theorem \ref{th_type2}}]
Since the POVM 
$\{\hat\Pi_0,\hat\Pi_1\}$
of 
$T^\inv_\alpha$
is given by 
$\hat\Pi_0
=(1-\alpha)\hat K_0$,
we have 
$
	\beta_{\hat\rho_{\vec\theta,0}^{\ox n}}[T^\inv_\alpha]
	=
	(1-\alpha)
	\beta_{\hat\rho_{\vec\theta,0}^{\ox n}}[T^\inv_0]
$.
Let 
$\vec u_m\in\Mat_\bR^{m,1}$
be the unit vector 
whose first entry is one.
For 
$\vec\theta\in\Mat_\bC^{m,1}$,
let 
$A\in\Ant_\bC^n$
be a solution to the equation 
$e^A\vec\theta=\|\vec\theta\|\vec u_m$.
Let 
$E_m\in\Mat_\bR^{m,n}$
be 
$
	(\overbrace{\vec  u_m,\vec u_m,...,\vec u_m}^n)
$.
By Lemma 
\ref{lem_g_a_b_z}, 
it holds that 
$
	\hat U_A
	\hat\rho_{\vec\theta,0}^{\ox n}
	\hat U_A^*
	=
	\hat\rho_{\|\vec\theta\|E_m,0}
$.
Let 
$\hat K'_0
\in\cL(\cH^{\ox n})$
be $\hat K_0$
for $m=1$.
Then,
it holds that 
$
	\Tr[\hat\rho_{\|\vec\theta\|E_m,0}\hat K_0]
	=
	\Tr[\hat\rho_{\|\vec\theta\|E_1,0}\hat K'_0]
	\Tr[\hat\rho_{0E_1,0}\hat K'_0]^{m-1}
$.
By Lemma
\ref{lem_type_ii_error_prob},
we have the statement.
\end{proof}
\donotdisplay
{
\[
	\beta_{T_\inv}
	=
	\frac{e^{-n\|\vec\theta\|^2}}{B(\frac{n-1}{2},\frac12)}
	\int_0^\pi
	e^{n\|\vec\theta\|^2\cos t}
	\sin^{n-2}t
	dt
	.
\]
}

\subsection{Proof of Theorem \ref{th_not_si} ($T^\HH_\alpha$ is not SI)}\label{sec_proof_not_si}

By Eq. (3.12) of 
Leonhardt \cite{leonhardt},
the Fourier transform of 
the Wigner function of 
$\hat\rho\in\cS(\cH)$ is 
$F_{\hat\rho}(u,v)=\Tr[\hat\rho\exp(-iu\hat q-iv\hat p)]$,
where 
$i=\sqrt{-1}$,
$\hat q=(\hat a+\hat a^*)/\sqrt2$
and
$\hat p=-i(\hat a-\hat a^*)/\sqrt2$.
For $k\in\{1,2,...,m\}$,
let 
$\hat q_k=(\hat a_k+\hat a_k^*)/\sqrt2\in\cL(\cH^{\ox m})$
and let
$\hat p_k=-i(\hat a_k-\hat a_k^*)/\sqrt2\in\cL(\cH^{\ox m})$.
Let 
\[
	\vec q=
	\begin{pmatrix}
	\hat q_1\\\vdots\\\hat q_m
	\end{pmatrix}
	\in\Mat_{m,1}^{m,1}
	,\
	\vec p=
	\begin{pmatrix}
	\hat p_1\\\vdots\\\hat p_m
	\end{pmatrix}
	\in\Mat_{m,1}^{m,1}
	\mbox{ and }
	\vec r
	=
	\begin{pmatrix}
	\vec q\\\vec p
	\end{pmatrix}
	\in\Mat_{m,1}^{2m,1}
	.
\]
Let 
\[
	\vec u=
	\begin{pmatrix}
	u_1\\\vdots\\ u_m
	\end{pmatrix}
	\in\Mat_\bR^{m,1}
	,\
	\vec v=
	\begin{pmatrix}
	v_1\\\vdots\\ v_m
	\end{pmatrix}
	\in\Mat_\bR^{m,1}
	\mbox{ and }
	\vec w
	=
	\begin{pmatrix}
	\vec u\\\vec v
	\end{pmatrix}
	\in\Mat_\bR^{2m,1}
	.
\]
The Fourier transform of the Wigner function of 
$\hat\rho\in\cS(\cH^{\ox m})$ is
$
	F_{\hat\rho}
	=
	\Tr[
		\hat\rho
		\exp(-i({^t}\vec w)\vec r)
	]
$.
Let 
$F_{\vec\theta,\eta,N}(\vec u,\vec v)
=F_{\hat\rho_{\vec\theta,\eta,N}}(\vec u,\vec v)$.
We use 
$\vec\mu_{\vec\theta}$ 
and 
$G_\eta$
defined in 
(\ref{eq_g_eta}).

\begin{lemma}\label{lem_fourier_wigner}
It holds that 
\[
	F_{\vec\theta,\eta,N}(\vec u,\vec v)
	=
	\exp\Big(
	-\frac{2N+1}{4}
	({^t}\vec w)G_\eta({^t}G_\eta)\vec w
	-i\sqrt2
	({^t}\vec w)G_\eta\vec\mu_{\vec\theta}
	\Big)
	.
\]
\end{lemma}
\begin{proof}
First, 
define 
$L\in\Mat_\bC^{2m}$
and 
$\dvec a\in\Mat_{m,1}^{2m,1}$
by 
\[
	L
	=
	\frac{1}{\sqrt2}
	\begin{pmatrix}
	I_m&iI_m\\
	I_m&-iI_m
	\end{pmatrix}
	\mbox{ and }
	\dvec a
	=L\vec r
	,
\]
respectively.
Let 
$\vec z
=-i(\vec u+i\vec v)/\sqrt2
\in\Mat_\bC^{m,1}$,
and let 
$\dvec z
={^t}({^t}\vec z,\vec z^*)
\in\Mat_\bC^{2m,1}$.
Let 
$\hat D_{\vec z}
=\exp(\dvec a^*K_m\dvec z)
\in\cU(\cH^{\ox m})$.
Then,
we have 
\[
	\exp(-i({^t}\vec w)\vec r)
	=
	\exp(-i({^t}\vec w)L^*\dvec a)
	=
	\exp(-i({^t}\vec w)L^*K_m^2\dvec a)
	=
	\hat D_{\vec z}
	.
\]
By Eq. (\ref{eq_d_ris}),
we have 
$\hat D_{\vec z}\Ket{\vec0_m}
=\Ket{\vec z}$.
By Eq. (\ref{eq_inner_product}),
we have 
$\big\langle\vec0_m
\Ket{\vec z}
=\exp(-\|\vec u\|^2/4-\|\vec v\|^2/4)$.
Hence, 
we obtain 
\begin{align}
	F_{\vec0_m,O_{2m},0}(\vec u,\vec v)
	=\exp(-({^t}\vec w)\vec w/4)
	.
	\label{eq_f000}
\end{align}
Next,
we have 
$L^*e^\eta L
=\exp(L^*\eta L)
=G_\eta$.
Moreover,
it holds that 
$F_{\vec\theta,\eta,0}
=
\Bra{\vec0_m}
\hat D_{\vec\theta}^*
\hat S_\eta^*
\exp(-i({^t}\vec w)L^*\dvec a)
\hat S_\eta
\hat D_{\vec\theta}
\Ket{\vec0_m}$.
By Lemma \ref{lem_s_rep},
it holds that 
$\hat S_\eta^*\dvec a\hat S_\eta
=e^\eta\dvec a$.
By Eq. (\ref{eq_displace}),
it holds that 
$\hat D_{\vec\theta}^*
\dvec a
\hat D_{\vec\theta}
=\dvec a+\sqrt2L\vec\mu_{\vec\theta}\hat I^{\ox m}$.
Hence,
we have 
\begin{align}
	F_{\vec\theta,\eta,0}(\vec u,\vec v)
	=&
	\Bra{\vec0_m}
	\exp[-i({^t}\vec w)L^*e^\eta(\dvec a+\sqrt2L\vec\mu_{\vec\theta})]
	\Ket{\vec0_m}
	\nonumber
	\\
	=&
	\Bra{\vec0_m}
	\exp[-i({^t}\vec w)G_\eta\vec r-i\sqrt2({^t}\vec w)G_\eta\vec\mu_{\vec\theta}]
	\Ket{\vec0_m}
	.
	\label{eq_f_not_000}
\end{align}
By (\ref{eq_f000})
and (\ref{eq_f_not_000}),
we obtain 
\begin{align}
	F_{\vec\theta,\eta,0}(\vec u,\vec v)
	=
	\exp\Big(
	-\frac14
	({^t}\vec w)G_\eta({^t}G_\eta)\vec w
	-i\sqrt2
	({^t}\vec w)G_\eta\vec\mu_{\vec\theta}
	\Big)
	.
	\label{eq_f_statement}
\end{align}
Next,
for 
$\vec x={^t}(x_1,x_2,...,x_m)\in\Mat_\bR^{m,1}$
and for 
$\vec y={^t}(y_1,y_2,...,y_m)\in\Mat_\bR^{m,1}$,
let 
$g_{\vec\theta,N}(\vec x,\vec y)
=(\pi N)^{-1}e^{-\|\vec x+i\vec y-\vec\theta\|^2/N}$.
Then, 
it holds that 
\begin{align}
	F_{\vec\theta,\eta,N}(\vec u,\vec v)
	=&
	\left.\overbrace{\idotsint}^{2m}\right._{\bR^{2m}}
	g_{\vec\theta,N}(\vec x,\vec y)
	F_{\vec x+i\vec y,\eta,0}(\vec u,\vec v)
	\prod_{k=1}^mdx_kdy_k
	.
	\label{eq_g_f}
\end{align}
By (\ref{eq_f_statement}),
the integrated factor of 
(\ref{eq_g_f})
is the Fourier transform of 
$g_{\vec\theta,N}(\vec x,\vec y)$.
Hence,
we obtain the statement.
\end{proof}

We use 
$\Sigma_{\eta,N}$
defined in (\ref{eq_vec_mu_vec_theta_eta_sigma_eta_n}).

\begin{lemma}\label{lem_multi_normal}
It holds that 
\[
	\frac{1}{\pi^m}
	\Bra{\vec z}\hat\rho_{\vec\theta,\eta,N}\Ket{\vec z}
	=
	\frac
	{\exp[-{^t}(\vec\mu_{\vec z}-G_\eta\vec\mu_{\vec\theta})\Sigma_{\eta,N}^{-1}(\vec\mu_{\vec z}-G_\eta\vec\mu_{\vec\theta})/2]}
	{(2\pi)^m\sqrt{\det[\Sigma_{\eta,N}]}}
	. 
\]
\donotdisplay
{
where 
$\vec\mu_{\vec\theta}$
and 
$G_\eta$ 
are defined in 
(\ref{eq_g_eta}),
and where 
}
\end{lemma}
\begin{proof}
Let 
$W_{\hat\rho}(x,y)$ be the Wigner function of 
$\hat\rho\in\cS(\cH)$.
By the overlap formula,
it holds that 
$\Tr[\hat\rho\hat\sigma]
=2\pi\iint_{\bR^2}
W_{\hat\rho}(x,y)W_{\hat\sigma}(x,y)dxdy$.
(See (3.22) of \cite{leonhardt}.)
Moreover,
by the Parseval's formula,
it holds that 
$\Tr[\hat\rho\hat\sigma]
=(2\pi)^{-1}\iint_{\bR^2}
\overline{F_{\hat\rho}(u,v)}F_{\hat\sigma}(u,v)dudv$.
Let 
$\vec\nu=\vec\mu_{\vec z}-G_\eta\vec\mu_{\vec\theta}$.
By Lemma 
\ref{lem_fourier_wigner},
we have 
\begin{align}
	\Bra{\vec z}\hat\rho_{\vec\theta,\eta,N}\Ket{\vec z}
	=&
	\frac{1}{(2\pi)^m}
	\left.\overbrace{\idotsint}^{2m}\right._{\bR^{2m}}
	\exp\Big(
	-({^t}\vec w)\Sigma_{\eta,N}\vec w
	-i\sqrt2({^t}\vec w)\vec\nu
	\Big)
	\nonumber
	\\
	&\times
	\prod_{i=k}^m
	du_kdv_k
	. 
	\label{eq_pure_hetero_probab}
\end{align}
Let 
$\vec w'=\vec w+i\Sigma_{\eta,N}^{-1}\vec\nu/\sqrt2$.
Then,
the argument of the exponential function 
of (\ref{eq_pure_hetero_probab})
is equal to 
$
	-({^t}\vec w')\Sigma_{\eta,N}\vec w'
	-({^t}\vec\nu)\Sigma_{\eta,N}^{-1}\vec\nu/2
$.
Calculating the Gausian intagration
and multiplying $\pi^{-m}$,
we obtain 
the statement.
\end{proof}

Let 
$\kappa
={^t}\vec\mu_{\vec\theta,\eta}\Sigma_{\eta,N}^{-1}\vec\mu_{\vec\theta,\eta}$.

\begin{lemma}\label{lem_norm}
For any $r>0$,
for any $\vec\theta\in\Mat_\bC^{m,1}\backslash\{\vec0_m\}$
and 
for any $N\ge0$,
there exists 
$\eta\in\Sqz^m$
such that 
\begin{align}
	\kappa
	=
	\frac
	{
	4r^2
	\|\vec\theta\|^2
	}
	{
	(2N+1)r^2+1
	}
	.
	\label{eq_norm}
\end{align}
\end{lemma}
\begin{proof}
Let 
$\zeta\in\Sqz^m$
be 
$
\log(r)
\begin{pmatrix}
O_m&I_m\\
I_m&O_m
\end{pmatrix}$.
Then,
$G_\zeta$
is 
$\begin{pmatrix}
rI_m&O_m\\O_m&r^{-1}I_m
\end{pmatrix}$.
Let 
$\vec u\in\Mat_\bR^{m,1}$
be any real column vector.
Then,
it holds that 
$\vec\mu_{\vec u}
=\vec\mu_{\vec u,O_{2m}}
=\begin{pmatrix}\vec u\\\vec0_m\end{pmatrix}
$,
and that 
\[
	{^t}\vec\mu_{\vec u,\zeta}\Sigma_{\zeta,N}^{-1}\vec\mu_{\vec u,\zeta}
	=
	r^2
	({^t}\vec u)
	\Big(
	r^2\frac{2N+1}{4}I_m+\frac14I_m
	\Big)^{-1}
	\vec u
	=
	\frac{4r^2\|\vec u\|^2}{r^2(2N+1)+1}
	.
\]
Let 
$A\in\Ant_\bC^m$
be a diagonal matrix 
satisfying 
$e^{-A}\vec\theta
\in\Mat_\bR^{m,1}$.
Let 
$\vec s
=e^{-A}\vec\theta$.
Let 
$U
=
\begin{pmatrix}
e^A&O_m\\O_m&e^{-A}
\end{pmatrix}$.
Let 
$\eta
\in\Sqz^m$
be 
$U\zeta U^*$.
There exists 
an orthogonal matrix 
$R\in\Mat_\bR^{2m}$
such that 
$G_\eta=RG_\zeta({^t}R)$
and 
${^t}R\vec\mu_{\vec\theta}
=\begin{pmatrix}\vec s\\\vec0_m\end{pmatrix}$.
Hence,
it holds that 
\begin{align*}
	\kappa
	=&
	{^t}(G_\zeta{^t}R\vec\mu_{\vec\theta})
	\Big(
	\frac{2N+1}{4}G_\zeta({^t}G_\zeta)+\frac14I_{2m}
	\Big)^{-1}
	(G_\zeta{^t}R\vec\mu_{\vec\theta})
	\\
	=&
	r^2({^t}\vec s)
	\Big(
	r^2\frac{2N+1}{4}I_m+\frac14I_m
	\Big)^{-1}
	\vec s
	.
\end{align*}
Hence,
we obtain 
(\ref{eq_norm}).
\end{proof}

\donotdisplay
{
Let $X$ and $Y$ be random variables 
or random vectors. 
If $X$ and $Y$
obey a common probability distribution,
then 
we write 
$X\stackrel p=Y$.
}

Let 
$\lambda$
be 
$n\kappa$.
Let 
$p_\lambda(f)$
be the probability density function of 
$F_{\mu,\nu;\lambda}$
defined in 
(\ref{eq_pdf_non_central_f}).
Let 
$c$ 
be the critical point of level $\alpha$,
that is,
the solution to 
$\int_c^\infty p_0(f)df=\alpha$.

\begin{lemma}\label{lem_compari_hh_0}
For any $m\ge1$,
for any $n\ge2m+1$,
for any $\vec\theta\in\Mat_\bC^{m,1}$,
for any $\eta\in\Sqz^m$,
for any $N\ge0$
and for any 
$\alpha\in(0,1)$,
there exists 
$\delta>0$
such that 
$\beta_{\hat\rho_{\vec\theta,\eta,N}^{\ox n}}[T^\HH_\alpha]
=1-\alpha-(1-\alpha-\delta)
\lambda/2
+o(\lambda)$
holds 
as $\lambda$
goes to zero.
\end{lemma}
\begin{proof}
Let 
$q(f)
=\partial p_\lambda(f)/\partial\lambda\big|_{\lambda=0}$.
It holds that 
\[
	q(f)
	=
	-\frac12
	p_0(f)
	+\frac12
	\frac{\mu+\nu}{\mu}
	\frac{\mu f}{\mu f+\nu}
	p_0(f)
	=
	\Big(
	\frac{\mu+\nu}{\mu f+\nu}
	f
	-1
	\Big)
	\frac{p_0(f)}{2}
	.
\]
Let 
$\delta
=(\mu+\nu)\int_0^c
(\mu f+\nu)^{-1}
f
p_0(f)df$.
Then,
it holds that 
$\int_0^cq(f)df
=(\delta-1+\alpha)/2$.
Hence,
we have 
$\beta_{\hat\rho_{\vec\theta,\eta,N}^{\ox n}}[T^\HH_\alpha]
=
1-\alpha
+(\delta-1+\alpha)\lambda/2+o(\lambda)$
as $\lambda\to0$.
\end{proof}

\begin{proof}[{\bf Proof of Theorem \ref{th_not_si}}]
(i)
By Lemma 
\ref{lem_norm},
$\lambda$ depends on 
$\eta$.
By Lemma \ref{lem_compari_hh_0}, 
$\beta_{\hat\rho_{\vec\theta,\eta,N}^{\ox n}}[T^\HH_\alpha]
=\int_0^c p_\lambda(f)df$
depends on $\eta$.
\\
(ii)
By Eq. (\ref{eq_pdf_non_central_f}),
$\beta_{\hat\rho_{\vec\theta,\eta,N}^{\ox n}}[T^\HH_\alpha]
=\int_0^cp_\lambda(f)df$
is greater than 
$e^{-\lambda/2}(1-\alpha)$.
By 
Lemma \ref{lem_norm},
for any $\varepsilon\in\bR$
with $0<\varepsilon<1-\alpha$,
there exists 
$\eta\in\Sqz^m$
such that 
${^t}\vec\mu_{\vec\theta,\eta}\Sigma_{\eta,N}^{-1}\vec\mu_{\vec\theta,\eta}
<2\log[(1-\alpha)/(1-\alpha-\varepsilon)]$,
that is,
$\beta_{\hat\rho_{\vec\theta,\eta,N}^{\ox n}}[T^\HH_\alpha]
>1-\alpha-\varepsilon$.
Hence,
we have 
$\sup_{\eta\in\Sqz^m}
\beta_{\hat\rho_{\vec\theta,\eta,N}^{\ox n}}[T^\HH_\alpha]
=1-\alpha$.
\end{proof}

\subsection{Proof of Theorem \ref{th_comparison}(Comparison of $\beta_{\hat\rho_{\theta,0}^{\ox3}}[T^\inv_\alpha]$ with $\beta_{\hat\rho_{\theta,0}^{\ox3}}[T^\HH_\alpha]$)}\label{sec_proof_comparison}

\donotdisplay
{
For a density operator 
$\hat\rho\in\cS(\cH^{\ox mn})$
and for a test $T$
of level $\alpha$,
the power function 
$\gamma_{\hat\rho}[T]$
is defined by 
$1-\alpha-
\beta_{\hat\rho}[T]$.
(See \cite{lehman_romano}.)
}

We compare orders 
of type II error probabilities 
for 
$\|\vec\theta\|\fallingdotseq0$
and for 
$\|\vec\theta\|\gg0$.

\begin{lemma}\label{lem_compari_inv_0}
For any $m\ge1$,
for any $n\ge2$
and for any $\alpha\in[0,1)$,
it holds that 
$\beta_{\hat\rho_{\vec\theta,0}^{\ox n}}[T^\inv_\alpha]
=1-\alpha-(1-\alpha)n\|\vec\theta\|^2
+o(\|\vec\theta\|^2)$
as $\|\vec\theta\|^2$ goes to zero.
\end{lemma}
\begin{proof}
Let 
$f(r)
=e^{-n r}\int_0^\pi
e^{n r\cos\varphi}
(\sin\varphi)^{n-2}
d\varphi$.
It holds that 
\[
	\frac{df}{dr}(0)
	=
	-n f(0)
	+
	n
	\int_0^\pi
	\cos\varphi
	(\sin\varphi)^{n-2}
	d\varphi
	=-nf(0)
	,
\]
and that 
$f(0)=B((n-1)/2,1/2)$.
By 
Theorem \ref{th_type2},
it holds that 
$\beta_{\hat\rho_{\vec\theta,0}^{\ox n}}[T^\inv_\alpha]
=(1-\alpha)f(\|\vec\theta\|^2)/B((n-1)/2,1/2)$.
Hence,
we have 
$\beta_{\hat\rho_{\vec\theta,0}^{\ox n}}[T^\inv_\alpha]
-1+\alpha
=-(1-\alpha)n\|\vec\theta\|^2
+o(\|\vec\theta\|^2)$
as 
$\|\vec\theta\|^2\to0$.
\end{proof}

\begin{lemma}\label{lem_compari_inv_inf}
If $m=1$,
$n=3$
and 
$N=0$,
then, 
for any $\alpha\in[0,1)$,
it holds that 
$\beta_{\hat\rho_{\theta,0}^{\ox3}}[T^\inv_\alpha]
=O(|\theta|^{-2})$
as $|\theta|$ goes to infinity.
\end{lemma}
\begin{proof}
Let 
$r=3|\theta|^2$,
let 
$s
=e^{-r}/B(1,1/2)$,
and let 
$t
=\int_0^\pi
e^{r\cos\varphi}\sin\varphi d\varphi$.
By Theorem 
\ref{th_type2},
it holds that 
$
	\beta_{\hat\rho_{\theta,0}^{\ox3}}[T^\inv_\alpha]
	=
	(1-\alpha)
	st
$.
We can calculate 
$t$ as 
\[
	t
	=
	\left.
	-\frac1r
	e^{r\cos\varphi}
	\right|_{\varphi=0}^\pi
	=
	\frac
	{e^{r}-e^{-r}}
	r
	.
\]
Hence,
we have 
$\beta_{\hat\rho_{\theta,0}^{\ox3}}[T^\inv_\alpha]
=(1-\alpha)(1-e^{-2r})/[rB(1,1/2)]
=O(|\theta|^{-2})$
as 
$\|\vec\theta\|^2\to\infty$.
\end{proof}

\donotdisplay
{
Hereafter,
consider the case of 
$\eta=O_{2m}$
and 
$N=0$.
Then,
$
	\lambda
	=
	n({^t}\vec\mu_{\vec\theta,\eta})\Sigma_{\eta,N}^{-1}\vec\mu_{\vec\theta,\eta}
$
is equal to 
$2n\|\vec\theta\|^2$.
Assume that 
$n\ge2m+1$.
Let 
$\mu=2m$
and let 
$\nu=n-\mu$.
Let 
$f$ be a random variable 
obeying 
the non-central $F$ distribution 
with $\mu$ and $\nu$ degrees of freedom 
and with 
non-centrality $\lambda$.
}

Let 
$
	\lambda
	=
	n({^t}\vec\mu_{\vec\theta,\eta})\Sigma_{\eta,N}^{-1}\vec\mu_{\vec\theta,\eta}
$. 
Let 
$p_\lambda(f)$
be the probability density function of 
$f$
given in (\ref{eq_pdf_non_central_f}).

\begin{lemma}\label{lem_compari_hh_inf}
For any $m\ge1$,
for any $n\ge2m+1$,
for any $\vec\theta\in\Mat_\bC^{m,1}$,
for any $\eta\in\Sqz^m$,
for any $N\ge0$
and for any $c>0$,
there exists $t>0$ such that 
$\int_0^cp_\lambda(f)df
=o(e^{-t\lambda})$
as $\lambda$ goes to infinity.
\end{lemma}
\begin{proof}
Let 
$A_0
=\int_0^cp_\lambda(f)df$.
Replace 
$\mu f/(\mu f+\nu)$
by $x$.
The Jacobian of this replacement is 
$df/dx
=(dx/df)^{-1}
=x^{-1}(1-x)^{-1}f$.
Let 
$b_0=\mu c/(\nu+\mu c)$.
Let 
\[
	q_k
	=
	\frac
	{e^{-\lambda/2}(\lambda/2)^k}{k!}
	\mbox{ and }
	r_k(x)
	=
	\frac
	{x^{k+\mu/2-1}
	(1-x)^{\nu/2-1}}
	{B(k+\mu/2,\nu/2)}
	.
\]
It holds that 
$A_0
=
\int_0^{b_0}
\sum_{k=0}^\infty
q_kr_k(x)dx$.
Choose 
$b_1\in(b_0,1)$
arbitrarily.
There exists 
$k_0\in\bN$
such that, 
for any $k\ge k_0$
and for any 
$x\in(0,b_0)$,
\begin{align}
	\frac
	{r_{k+1}(x)}
	{r_k(x)}
	=
	\frac
	{k+(\mu+\nu)/2}
	{k+\mu/2}
	x
	<
	\frac
	{k+(\mu+\nu)/2}
	{k+\mu/2}
	b_0
	<
	b_1
	\label{eq_b_1}
\end{align}
holds.
Let 
$R_k
=\int_0^{b_0}r_k(x)dx$.
By (\ref{eq_b_1}),
for any $k\ge k_0$,
it holds that 
$
	R_{k+1}<b_1R_k
$.
Let 
$A_1
=\sum_{k=0}^{k_0-1}
q_kR_k$,
and let 
$A_2
=A_0
-A_1
=\sum_{k=k_0}^\infty
q_kR_k$.
It holds that 
\begin{align}
	A_2
	<
	\sum_{k=k_0}^\infty
	q_kR_{k_0}b_1^{k-k_0}
	<
	\sum_{k=0}^\infty
	q_kR_{k_0}b_1^{k-k_0}
	=
	b_1^{-k_0}
	R_{k_0}
	e^{\lambda(b_1-1)/2}
	.
	\label{eq_a_2}
\end{align}
For any $s>0$,
it holds that 
\begin{align}
	A_1
	<
	\sum_{k=0}^{k_0-1}
	q_k
	e^{(k_0-1-k)s}
	<
	\sum_{k=0}^\infty
	q_k
	e^{(k_0-1-k)s}
	=
	e^{(k_0-1)s}
	e^{\lambda(e^{-s}-1)/2}
	.
	\label{eq_a_1}
\end{align}
Let 
$t
=
2^{-1}
\min\{
(1-b_1)/2,
(1-e^{-s})/2
\}$.
By (\ref{eq_a_2})
and (\ref{eq_a_1}),
we have 
$A_0
=o(e^{-t\lambda})$
as 
$\lambda\to\infty$.
\end{proof}

\begin{proof}[{\bf Proof of Theorem \ref{th_comparison}}]
(i)
By Lemmas 
\ref{lem_compari_inv_0}
and 
\ref{lem_compari_hh_0},
there exists 
$\delta>0$
such that 
$
\beta_{\hat\rho_{\theta,0}^{\ox3}}[T^\HH_\alpha]
-
\beta_{\hat\rho_{\theta,0}^{\ox3}}[T^\inv_\alpha]
=3\delta|\theta|^2
+o(|\theta|^2)$
holds 
as $|\theta|\to0$.
\\
(ii)
By Lemmas 
\ref{lem_compari_inv_inf}
and 
\ref{lem_compari_hh_inf},
we obtain the statement.
\end{proof}

\section*{Acknowledgements}

The author thanks 
Prof.
Keiji Matsumoto of NII 
and 
Prof.
Fuyuhiko Tanaka of Osaka University 
for comments and suggestions.
\donotdisplay
{
In particular,
Lemma \ref{lem:2_realization}
was suggested by 
Matsumoto,
and 
Lemma \ref{lem_tanaka}\label{lem_f_k}
was suggested by 
Tanaka.
}

\end{document}